\documentclass[a4paper,english,journal,final]{IEEEtran}
\usepackage{epsfig,psfrag}
\usepackage{amsmath,amsfonts,amssymb,amsxtra,bm}
\usepackage{tabularx}
\usepackage{esint}

\usepackage[T1]{fontenc}
\usepackage[latin9]{inputenc}
\usepackage{color}
\usepackage{graphicx}
\usepackage{esint}
\usepackage{babel}
\usepackage{cite}

 {\usepackage{color}}{}
\usepackage{cite}

\newcommand{\highlight}[1]{#1}

\newcommand{\ist}{\hspace*{.3mm}}
\newcommand{\rmv}{\hspace*{-.3mm}}
\renewcommand{\aa}{\lambda}
\newcommand{\bb}{\nu}

\newcommand{\slab}{0.6}	% small label
\newcommand{\mlab}{0.7}	% medium label
\newtheorem{theorem}{Theorem}

\makeatother

\begin{document}

\title{Cooperative Synchronization in Wireless Networks}

\author{\IEEEauthorblockN{Bernhard Etzlinger,~\IEEEmembership{Student Member,~IEEE,}
Henk Wymeersch,~\IEEEmembership{Member,~IEEE,} \\
and Andreas Springer,~\IEEEmembership{Member,~IEEE}} %
\thanks{B. Etzlinger and A. Springer are with the Institute of Communications
Engineering and RF-Systems, Johannes Kepler University, Linz, Austria,
e-mail: \{b.etzlinger, a.springer\}@nthfs.jku.at. H. Wymeersch is
with the Department of Signals and Systems, Chalmers University of
Technology, Gothenburg, Sweden. email: henkw@chalmers.se. 

This research was supported, in part, by the European Research Council,
under Grant No. 258418 (COOPNET) and by the COMET K2 Center ``Austrian
Center of Competence in Mechatronics (ACCM)''. The COMET Program
is funded by the Austrian Federal government, the Federal State of
Upper Austria and the Scientific Partners of ACCM.%
} }
\maketitle
\begin{abstract}
\highlight{Synchronization is a key functionality in wireless network, enabling a wide variety of services. We consider a Bayesian inference framework whereby network nodes can achieve phase and skew synchronization in a fully distributed way. In particular, under the assumption of Gaussian  measurement noise, we derive two message passing methods (belief propagation and mean field), analyze their convergence behavior, and perform a qualitative and quantitative comparison with a number of competing algorithms. We also show that both methods can be applied in networks with and without master nodes. Our performance results are complemented
by, and compared with, the relevant Bayesian Cramér--Rao bounds.}

\end{abstract}
\begin{IEEEkeywords}
Network synchronization, belief propagation, mean field, distributed estimation,
Bayesian Cramér--Rao bound. 
\end{IEEEkeywords}

\section{Introduction}

\label{sec:introd} 
%--------------------------------------------------------------------
% SECTION: Introduction
%--------------------------------------------------------------------

\PARstart{W}{ireless} \highlight{ networks (WNs)
%are faced to emerging tasks,
must deliver a wide variety of services,
assuming a decentralized, but cooperative operation of the WN.
Many of these 
%tasks
services 
have stringent requirements on time alignment among the nodes in the WN, while each node has its individual local clock. Usually these clocks are counters driven by a local oscillator, and differences appear in counter offsets and in oscillator frequencies.
Time alignment is necessary, e.g., for cooperative transmission and distributed beamforming \cite{jagannathan}, time-division-multiple-access
communication protocols \cite{demirkol}, 
%and 
duty-cycling \cite{ganeriwal05},
%, which require
%low offsets in the clocks for a specified time instant. 
and 
localization and tracking methods
\cite{hlinka,elson2}, or location based control schemes \cite{antonelli13}.
%depend on distance estimates, often obtained by round-trip time measurements. Therefore, 
%time intervals need to be accurately represented, resulting in strict frequency alignment requirements.
In these tasks, the representation of specific time instants requires low clock offsets, while for accurate representation of time intervals, strict frequency alignment is necessary.
To guarantee correct operation,
%of these
%before mentioned cooperative 
%tasks, 
the clocks need to be aligned up to a certain application-specific accuracy.
% by synchronization, where the accuracy requirement is determined by the specific application.
}

Synchronization is a widely studied topic. 
Existing network synchronization schemes differ mainly in how local time information
is encoded, exchanged, and processed \cite{simeone}. In this work, we will limit
ourselves to so-called packet-coupled synchronization, whereby local
time is encoded in time stamps and exchanged via packet transmissions
\cite{wu}. 
\highlight{
%Traditional 
Commonly used algorithms, which consider both offset and frequency synchronization, are the {Reference Broadcast Synchronization} (RBS) \cite{elson} and the {Flooding Time Synchronization Protocol} (FTSP) \cite{maroti}. 
%FTSP requires a time reference introduced by a master node (MN). 
Both methods require a
specified network structure and do not perform  synchronization in a distributed manner. This increases communication and computation overhead to maintain the structure, makes the network more vulnerable to node failures, and reduces the scalability.
More recent synchronization algorithms work fully distributed, and are well suited 
%for 
to 
cooperative networks. For offset and frequency estimation, there are methods based on consensus \cite{zennaro, schenato,maggs,zhao}, and gradient descent \cite{solis}. They typically suffer from slow convergence speed, and thus require the exchange of a high number of data packets in the network to achieve a desired accuracy.
}
\highlight{
Recently, distributed Bayesian estimators were proposed, which provide a maximum a posteriori (MAP) estimate using belief propagation (BP) on factor graphs (FG). The successful application of offset synchronization in \cite{leng,zennaro13} showed superior estimation accuracy and higher convergence rate than competing distributed algorithms, in the case when a master node (MN) with reference time is available. The extension to joint offset and frequency synchronization is not straightforward, as nonlinear dependencies are introduced in the measurement model and thus in the likelihood of the measurements. In parallel to this work, \cite{du13} proposed such 
an extension by modifying the measurement equations. This modification restricts measurement model and results in an auxiliary function rather than a likelihood function. Thus, no MAP solution is obtained.
% an extension by replacing the likelihood of the observations with an auxiliary function,
%achieved by a modification of the measurement equations. The used modification 
% which 
% puts restrictions on the gathering of time measurements, and 
% %the achieved solution is not 
% does not lead to 
% a MAP estimate. 
%All existing FG based works apply BP, which has the disadvantage of a high computational complexity.
A general drawback of BP in \cite{leng,du13} is the high computational complexity, scaling quadratically in the number of neighbors.
}

% \highlight{ 
% As performance bounds a number of Cramér--Rao type bounds were introduced. In recent work, a Bayesian Cramér--Rao bound (BCRB) was considered for phase synchronization \cite{ahmad}. The Cramér--Rao bounds used for the phase and skew model in \cite{du13, rajan} are based on technically not correct likelihood functions.
% }
In this paper, we build on the 
work from \cite{leng},
%(BP phase synchronization with MN), 
considering
both relative clock \highlight{phases (offsets)} and clock \highlight{frequencies (skews)}. Our contributions
are as follows: 
\highlight{
\begin{itemize}
  \item Based on a measurement model from experimental data, we derive an approximate, yet accurate 
	statistical model that allows a Gaussian reformulation of the MAP estimation of the clock parameters.
	% enabling the development of computationally simple algorithms; 
  \item We propose a BP and a mean field (MF) message passing algorithm based on the statistical model.
	When MNs are available, MF provides highly accurate synchronization with low computational complexity. To the best of our knowledge, this is the first application of MF to the network synchronization problem.
  \item %By analyzing the convergence, we show that both algorithms can perform network synchronization with and without MNs for time reference.
	We provide convergence conditions for BP and MF synchronization with and without MNs.
  \item We derive a Bayesian Cramér--Rao bound (BCRB), which serves as a fundamental
	performance bound for the case when prior information on the local
	clock parameters is available.\footnote{\highlight{In recent work, %a phase only BCRB was considered \cite{ahmad}, while 
	a non-Bayesian CRB for joint phase and skew estimation was derived in \cite{du13, chepuri13, zennaro13}.}}
	% on a different likelihood function.
\end{itemize}
}

The remainder of this paper is organized as follows: In Section
\ref{sec:sysmod}, we introduce the clock model, the network model,
and the measurement protocol. In Section \ref{sec:stateoftheart},
we provide an overview of the state-of-the-art on distributed synchronization,
under both phase and skew uncertainties. Exact and simplified statistical
models of measurement likelihoods and clock priors are derived
in Section \ref{sec:statmod}, followed by a description of the BCRB
in Section \ref{sec:BCRB}. The simplified models from Section \ref{sec:statmod}
are used to derive two Bayesian algorithms based on message passing
in Section \ref{sec:MP}. In Section \ref{sec:NumAn}, we numerically
study the properties of the proposed algorithms and compare them to
state-of-the-art algorithms from Section \ref{sec:stateoftheart}.
Finally, we draw our conclusions in Section \ref{sec:Conclusion}.

%--------------------------------------------------------------------
% SECTION: System Model and Optimization Criteria
%--------------------------------------------------------------------

\section{System Model}

\label{sec:sysmod}

\subsection{Network Model }

\begin{figure}
\centering{}\includegraphics[width=0.5\columnwidth]{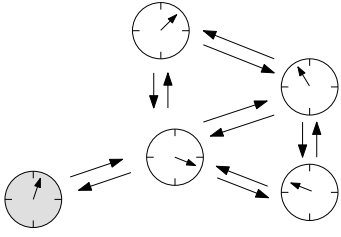}
\caption{\label{fig:connectivity}Connectivity graph of a wireless network
with $M=1$ MN (shaded) and $A=4$ agent nodes.}
\end{figure}

We consider a static network 
comprising a set $\mathcal{M} \triangleq \{1,\ldots,M\}$ of MNs and a set $\mathcal{A} \triangleq \{M+1,\ldots,N\}$
of agent nodes (ANs) (see Fig.~\ref{fig:connectivity}). The $M=|\mathcal{M}|$
fully synchronous MNs impose a common time reference to the
network. The $A=|\mathcal{A}|$ ANs have imperfect clocks
that may not run synchronously with the reference time. 

The topology is defined by the communication set $\mathcal{C} \rmv\subseteq\rmv \mathcal{I} \!\times\rmv \mathcal{I}$.
If two nodes $i,j \in \mathcal{I}$ can communicate, then  $(i,j) \!\in\! \mathcal{C}$ and $(j,i) \!\in\! \mathcal{C}$.
Connections among MNs are not considered, i.e., $(i,j) \notin \mathcal{C}$ if $i,j \!\in\! \mathcal{M}$.
For each $i \!\in\! \mathcal{A}$ we define a neighborhood set $\mathcal{T}_i \subseteq \mathcal{I} \rmv\setminus\! \{i\}$ that includes all 
$j \!\in\! \mathcal{I}$ that communicate with $i$, i.e., $j \!\in\! \mathcal{T}_i$ if and only if $(i,j) \!\in\! \mathcal{C}$.

The network is assumed to be connected, so that there is a path
between every pair of nodes.

\subsection{Clock Model}

Each network node $i$ possesses a clock displaying local time $c_{i}(t)$,
related to the reference time $t$ by 
\begin{equation}
c_{i}(t)=\alpha_{i}t+\beta_{i},\label{eq:local_clock}
\end{equation}
%where $\beta_{i}$ and $\alpha_{i}$ are the clock phase
% and the clock skew of node $i$, respectively \cite{wu}.
where $\beta_{i}$ is the clock phase of node $i$ and $\alpha_{i}$
is the clock skew of node $i$ \cite{wu}. When $i\rmv\in\rmv\mathcal{M}$,
$\alpha_{i}\!=\!1$ and $\beta_{i}\!=\!0$. When $i\rmv\in\rmv\mathcal{A}$, \linebreak both
$\alpha_{i}$ and $\beta_{i}$ are considered as random variables.
The clock phase $\beta_{i}$ depends on the initial network state,
and can be modeled with an uninformative prior (e.g., as uniformly
distributed over a large range, or, equivalently, having Gaussian
distribution with a large variance $\sigma_{\beta,i}^2$ \cite{leng}). The clock
skew $\alpha_{i}$ depends on the quality of the clocks, typically
expressed in parts per million (ppm), and is modeled as a Gaussian
random variable \cite{wu2} with mean 1 and variance $\sigma_{\alpha,i}^{2}$. %F. Ferrari, A. Meier, and L. Thiele, ``Accurate clock models for simulating wireless sensor networks,'' OMNeT++, Mar. 2010.
Nodes with more sophisticated clocks will have smaller $\sigma_{\alpha,i}^{2}$.
\highlight{Note that in reality the clock skews $\alpha_i$ are not static over time, as they change with ambient environment variations (e.g., temperature). Such variations are typically much slower than the update rate of synchronization protocols, and can thus safely be ignored.}

The following notation will be convenient: $\bm{\theta}_{i}=[\alpha_{i},\beta_{i}]^{\mathrm{T}}$,
$\bm{\vartheta}_{i} = [\lambda_i, \bb_i]^{\mathrm{T}} = [1/\alpha_{i},\beta_{i}/\alpha_{i}]^{\mathrm{T}}$
.
%The following notation will be convenient: $\bm{\theta}_{i}=[\alpha_{i},\beta_{i}]^{\mathrm{T}}$,
%$\bm{\theta}_{ij}=[\bm{\theta}_{i}^{\mathrm{T}},\bm{\theta}_{j}^{\mathrm{T}}]^{\mathrm{T}}$,
%$\bm{\theta}=[\bm{\theta}_{1}^{\mathrm{T}},\ldots,\bm{\theta}_{A}^{\mathrm{T}}]^{\mathrm{T}}$,
%$\bm{\vartheta}_{i}=[1/\alpha_{i},\beta_{i}/\alpha_{i}]^{\mathrm{T}}$,
%$\bm{\theta}_{ij}'=[\left(\bm{\vartheta}_{i}\right)^{\mathrm{T}},\left(\bm{\vartheta}_{j}\right)^{\mathrm{T}}]^{\mathrm{T}}$,
%and $\bm{\vartheta}=[\left(\bm{\vartheta}_{1}\right)^{\mathrm{T}},\ldots,\left(\bm{\vartheta}_{A}\right)^{\mathrm{T}}]^{\mathrm{T}}$. 

\subsection{Measurement Model}

\begin{figure}
\centering{}\includegraphics[width=0.9\columnwidth]{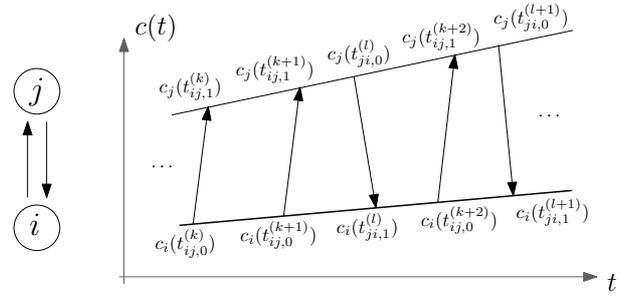}
\caption{\label{fig:2_way_msg}
Local clock counters $c_i(t)$ and $c_j(t)$ of node $i$ and $j$ w.r.t 
a reference time $t$, and recorded time stamps of the corresponding 
asymmetric packet exchange.}
\end{figure}

Following the asymmetric modeling in \cite{chepuri13}, which is an extension to \cite{wu, leng, zennaro13, du13, noh}, node pairs $(i,j)\in\mathcal{C}$ exchange
packets with time stamps to measure their local clock parameters $\bm{\theta}_i, \bm{\theta}_j$.
Node $i$ transmits $K_{ij}\ge1$ packets to node $j$ and 
node $j$ transmits
%receives
$K_{ji}\ge1$ packets
%from node $j$
to node $i$.
The $k$th ``$i \rmv\to\! j$'' packet 
%% transmitted from node $i$ to node $j$ 
(where $k \in \{1,\ldots,K_{ij}\}$) leaves node $i$ at time $t_{ij,0}^{(k)}$ and arrives at node $j$ 
after a delay $\delta_{ij}^{(k)}$, at measured time 
\begin{equation}
t_{ij,1}^{(k)}=t_{ij,0}^{(k)}+\delta_{ij}^{(k)}.\label{eq:basicrelation}
\end{equation}
The delay $\delta_{ij}^{(k)}$ is expressed in true time and can be broken up 
as $\delta_{ij}^{(k)}=\Delta_{ij}+w_{ij}^{(k)}$
\cite{maroti}, where $\Delta_{ij}$ is a deterministic component
(related to coding and signal propagation) and $w_{ij}^{(k)}$ is
a stochastic component.
% \footnote{The delays have multiple sources, as described 
% in \cite{maroti}. At this level, system design goals are to minimize non-deterministic 
% parts and to achieve highly accurate time-stamps, as done by close to physical layer time 
% stamping \cite{mahmood09}. Such mechanisms eliminates multiple delays, such as caused by routing and queuing.}
The nodes 
%have only access
record %to
$c_{j}\rmv(t_{ij,1}^{(k)})\rmv$ and $c_{i}(t_{ij,0}^{(k)})$,
which can be related to (\ref{eq:basicrelation}) through (\ref{eq:local_clock})
as 
\begin{equation}
c_{j}(t_{ij,1}^{(k)}) \,=\, \psi^{(k)}_{i\to j}(\bm{\theta}_{i},\bm{\theta}_{j},\Delta_{ij}) \ist+\ist w_{ij}^{(k)}\alpha_{j} \,,
\label{eq:basicmeasurement}
\vspace*{-1.5mm}
\end{equation}
with the deterministic part
\vspace*{-1mm}
\begin{equation}
\psi^{(k)}_{i\to j}(\bm{\theta}_{i},\bm{\theta}_{j},\Delta_{ij}) \,\triangleq\, \frac{c_{i}(t_{ij,0}^{(k)}) \rmv-\rmv \beta_{i}}{\alpha_{i}}\alpha_{j} 
  \ist+\ist \beta_{j} \ist+\ist \Delta_{ij} \alpha_{j} \,.
\label{eq:basicmeasurement_1}
\end{equation}
A similar relation holds for the packets
sent by node $j$ to node $i$, by exchanging $i$ and $j$ in (\ref{eq:basicmeasurement}).
The aggregated 
%% observation (measurements) 
measurement of
%% between 
nodes $i$ and $j$ 
%node $i$ w.r.t. $j$ 
is thus given by $\mathbf{c}_{ij} \triangleq [\mathbf{c}_{i \to j}^\mathrm{T} \,\ist\ist \mathbf{c}_{j \to i}^\mathrm{T}]^\mathrm{T}\!$,
with $\mathbf{c}_{i\to j} \triangleq \big[c_{j}(t_{ij,1}^{(1)}) \ist\cdots\, c_{j}(t_{ij,1}^{(K_{ij})}) \big]^{\mathrm{T}}\!$ and 
$\mathbf{c}_{j\to i} \triangleq \big[c_{i}(t_{ji,1}^{(1)})$\linebreak %%%%%%%%
$\cdots\, c_{i}(t_{ji,1}^{(K_{ji})}) \big]^{\mathrm{T}}\!$.
For later use, we also define the 
\vspace{-.4mm}
(recorded, not measured) time stamp vectors
$\tilde{\mathbf{c}}_{i\to j} \rmv\triangleq\rmv \big[c_{i}(t_{ij,0}^{(1)}) \cdots c_{i}(t_{ij,0}^{(K_{ij})}) \big]^{\mathrm{T}}\!$ and 
 $\tilde{\mathbf{c}}_{j\to i} \triangleq \big[c_{j}(t_{ji,0}^{(1)}) \ist\cdots\, c_{j}(t_{ji,0}^{(K_{ji})}) \big]^{\mathrm{T}}\!$.

\begin{figure}
\begin{centering}
% text strings:
\psfrag{s01}[t][t][\mlab]{\color[rgb]{0,0,0}\setlength{\tabcolsep}{0pt}\begin{tabular}{c}Transmission delay [$\mu s$]\end{tabular}}%
\psfrag{s02}[b][b][\mlab]{\color[rgb]{0,0,0}\setlength{\tabcolsep}{0pt}\begin{tabular}{c}Relative occurrence\end{tabular}}%
\psfrag{s05}[l][l][\slab]{\color[rgb]{0,0,0}Gaussian fit}%
\psfrag{s06}[l][l][\slab]{\color[rgb]{0,0,0}Measurement}%
\psfrag{s07}[l][l][\slab]{\color[rgb]{0,0,0}Gaussian fit}%
\psfrag{s09}[][]{\color[rgb]{0,0,0}\setlength{\tabcolsep}{0pt}\begin{tabular}{c} \end{tabular}}%
\psfrag{s10}[][]{\color[rgb]{0,0,0}\setlength{\tabcolsep}{0pt}\begin{tabular}{c} \end{tabular}}%
%
% xticklabels:
\psfrag{x01}[t][t][\slab]{7.4}%
\psfrag{x02}[t][t][\slab]{7.6}%
\psfrag{x03}[t][t][\slab]{7.8}%
\psfrag{x04}[t][t][\slab]{8}%
\psfrag{x05}[t][t][\slab]{8.2}%
\psfrag{x06}[t][t][\slab]{8.4}%
\psfrag{x07}[t][t][\slab]{8.6}%
%
% yticklabels:
\psfrag{v01}[r][r][\slab]{0}%
\psfrag{v02}[r][r][\slab]{0.02}%
\psfrag{v03}[r][r][\slab]{0.04}%
\psfrag{v04}[r][r][\slab]{0.06}%
\psfrag{v05}[r][r][\slab]{0.08}%
\psfrag{v06}[r][r][\slab]{0.1}%
\includegraphics[width=1\columnwidth]{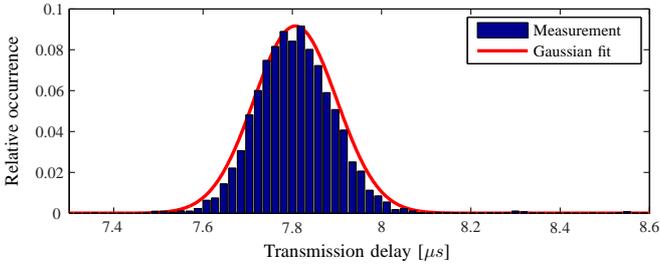} 
\caption{Measurement data and Gaussian fit of the delay $t_{ij,1}^{(k)} - t_{ij,0}^{(k)} = \delta_{ij}^{(k)}=\Delta_{ij}+w_{ij}^{(k)}$.}
\label{fig:delay_spread}
\par\end{centering}
\end{figure}

\vspace{0.3mm}
We model $\Delta_{ij}=T_{c}+T_{f,ij}$ as comprising
a hardware related computation time $T_{c}$, and a time of flight
$T_{f,ij}$. We further suppose that $\Delta_{ij}=\Delta_{ji}$.
Based on the results of a measurement campaign, shown in Fig.~\ref{fig:delay_spread},
with two Texas Instruments ez430-RF2500 evaluation boards\footnote{We placed
the boards 1 meter apart and transmitted 10,000 packets, collecting the corresponding transmit $t_{ij,0}^{(k)}$ and receive times $t_{ij,1}^{(k)}$.%
Via the general debug output (GDO with GDOx\_CFG = 6) of the CC2500
transceiver chip, time of transmission and time of reception was measured.
},
we modeled $w_{ij}^{(k)} \!\!\sim\! \mathcal{N}(0,\sigma^2_w)$, which is congruent with the models from \cite{noh,du13,leng,wu}. 
\highlight{The evaluated signal corresponds to a time stamping close to the physical layer \cite{loschmidt}, also often referred to as MAC time stamping \cite{maroti, zhao, maggs}. Using this concept, nondeterministic delays from higher layers, such as routing and queuing delays, are eliminated, and the parameters of the delay distribution are assumed to be static}.

\subsection{Network Synchronization}

Our goal is to infer the local clock parameters\footnote{
\highlight{
If synchronization would only correct the offset values, the existing frequency mismatches cause a drift of these offsets over time. This requires frequent resynchronization, leading to higher energy consumption. Also, frequency requirements of the application can only be achieved by using expensive hardware. Using \emph{joint} offset and frequency synchronization, frequency requirements can be met with less expensive hardware and resynchronization intervals can be increased.
}
}$\alpha_{i}$ and
$\beta_{i}$ (or an invertible transformation thereof), based on the
measurements and the prior clock information. In the following,
sections, we will describe standard approaches to solve this problem,
followed by our proposed Bayesian approach. 

%--------------------------------------------------------------------
% SECTION: State of the Art
%--------------------------------------------------------------------

\section{State-of-the-Art}

\label{sec:stateoftheart} 
%--------------------------------------------------------------------
% SECTION: State of the art
%--------------------------------------------------------------------

In this section, we briefly present a selection of existing synchronization methods
for the presented clock and network model. We limit our overview to
algorithms that perform 
%parameter estimation 
skew and phase synchronization 
based on time stamp
exchange in a fully distributed manner, where every node runs the
same algorithm. Within this class, we discuss approaches based on consensus \cite{schenato}, alternating direction of multiplier
method (ADMM) \cite{zennaro}, and loop constrained combination of
pairwise estimations \cite{solis}.
%\highlight{ Additionally, we provide coarse complexity estimate $C$, where 
%the operations $+,-,*,\sqrt{\cdot},\log$ and $\exp$ are equated with by one operation cost $O$.}

\subsection{Average TimeSync}

Consensus protocols are based on averaging information received
from neighbors, and thus have low computational complexity. Moreover,
%MNs, if present, may lead to a decrease
%in convergence speed. 
MNs are not considered. A reference time can be introduced, if a single node does not update its local parameters. This modification leads to a decreased convergence speed.
In the Average TimeSync (ATS) algorithm from \cite{schenato},
every node has a virtual clock 
\[
\hat{c}_{i}(c_{i}(t))=\hat{\alpha}_{i}c_{i}(t)+\hat{\beta}_{i}=\hat{\alpha}_{i}\alpha_{i}t+\hat{\alpha}_{i}\beta_{i}+\hat{\beta}_{i},
\]
which is controlled by a virtual skew $\hat{\alpha}_{i}$ and a virtual
phase $\hat{\beta}_{i}$. By adjusting the virtual skew and phase,
ATS assures asymptotic agreement on the virtual clocks $\lim\limits _{t\rightarrow\infty}\hat{c}_{i}(c_{i}(t))=\tau_{v}(t)$,
$\forall i\in\mathcal{A}$, where $\tau_{v}(t)$ is a network-wide
common time. The algorithm assumes $\delta_{ij}=0$. 
%\highlight{Per parameter update involving all neighbors, each node needs to broadcast one packet, 
%and perform $C \approx 21 \, | \mathcal{T}_i | \, O$.}

\subsection{ADMM Consensus}

In ADMM consensus from \cite{zennaro}, relative skews $\alpha_{ij}=\alpha_{i}/\alpha_{j}$
and relative phase offsets $\beta_{ij}=\beta_{i}-\beta_{j}$ are assumed
to be available a priori (e.g., \highlight{from a phase locked loop (PLL),} or from an estimation algorithm such as
\cite{noh}). Then follows a network-wide correction of the local
clock parameters in discrete instances $k\Delta T$. Collecting the
clock skews in $\mathbf{T}^{(k)}=[\alpha_{1}^{(k)}\Delta T,\ldots,\alpha_{A}^{(k)}\Delta T]^{\mathrm{T}}$,
and the clock phases in $\bm{\beta}^{(k)}=[\beta_{1}^{(k)},\ldots,\beta_{A}^{(k)}]^{\mathrm{T}}$,
control signals $\mathbf{u}^{(k)}$ and $\mathbf{v}^{(k)}$ are applied
as 
\begin{align*}
\bm{\beta}^{(k+1)} & =\bm{\beta}^{(k)}+\mathbf{T}^{(k)}+\mathbf{u}^{(k)}\\
\mathbf{T}^{(k+1)} & =\mathbf{T}^{(k)}+\mathbf{v}^{(k)}.
\end{align*}
The computation of $v_{i}^{(k)}$ at a node $i$ is based on ADMM
and requires knowledge of $\alpha_{jk}$ for all nodes $j\in\mathcal{T}_i,k\in\text{N}(j)$,
i.e., from all two-hop neighbors. It can be shown that as $k\to+\infty$,
$\mathbf{T}^{(k)}\to \bar{T} \cdot \mathbf{1}_{A,1}$ for some common value $\bar{T}$, where $\mathbf{1}_{k,l}$ 
denotes a $k\times l$ matrix with all entries equal to one.
Finally, agreement on the clock phases is achieved through the control
signal $\mathbf{u}^{(k)}$ using average consensus. Although fully
distributed and master-free, ADDM consensus requires $I$ inner iterations for offset 
measurements and outer iterations for the consensus. It further relies 
on a step size parameter, whose optimal value depends on global network properties. 
%\highlight{Using $K_n$ measurements for the pairwise estimates according to \cite{noh}, the algorithm requires 
%$K_n + 2 \, I$ broadcasts and $C \approx  | \mathcal{T}_i | \, (8 K_n + 2 I + 12) \, O$ per parameter update.}

\subsection{Loop Constrained Synchronization}

In \cite{solis}, it was observed that for every closed loop $\mathcal{L}$
in the network, it is such that $\sum_{i,j\in\mathcal{L}}\tilde{x}_{ij}=0,$
for $\tilde{x}_{ij}=\beta_{i}-\beta_{j}$ and $\tilde{x}_{ij}=\log(\alpha_{i}/\alpha_{j})$.
 Using these constraints, the absolute clock values are determined
via coordinate descent of the least squares problem \cite{solis} 
\begin{align*}
\hat{\mathbf{v}}=\arg\min_{\mathbf{v}}\|\mathbf{A} \mathbf{v}-\tilde{\mathbf{x}}\|^{2},
\end{align*}
where $\mathbf{A}$ is the incidence matrix representing a directed topology, $\mathbf{v}$
the vector of absolute clock parameters (phase or skew) and $\tilde{\mathbf{x}}$
is the collection of offset measurements. In order to find a global
optimum, a MN needs to be selected. During the iterations, local
estimates on absolute skew and phase are exchanged with all one-hop
neighbors.
%\highlight{In total, the method requires 2 broadcasts and computations of $C \approx 29 \, | \mathcal{T}_i | \, O$ per parameter update.}

%--------------------------------------------------------------------
% SECTION: Statistical Models
%--------------------------------------------------------------------

\section{Statistical Models}

\label{sec:statmod} 
%--------------------------------------------------------------------
% SECTION: Statistical Model
%--------------------------------------------------------------------

The above-mentioned algorithms are all non-Bayesian, and thus do not
fully exploit all statistical information present in the network.
When clock skews are known, fast, distributed Bayesian algorithms
were derived in \cite{leng}. When clock skews are unknown, the naive
extension of \cite{leng} would lead to impractical algorithms, due
to the complex integrals that need to be computed. In this section,
we propose a series of approximations to measurement likelihoods and
prior distributions, with the aim of a simple representation of the 
posterior distribution. Using this simplifications, the maximum a posteriori (MAP) estimate of the clock 
parameters (or a transformation thereof) can be found with reasonable
complexity.

\subsection{Likelihood Function\label{sub:Likelihood-Function}}

Because of \eqref{eq:basicmeasurement} and the statistical properties of $w_{ij}^{(k)}$,
%% fact that $w_{ij}^{(k)} \!\sim\rmv \mathcal{N}(0,\sigma^2_w)$, 
%% Via the statistical properties of the measurement noise and the model in \eqref{eq:basicmeasurement}, 
the local likelihood function of nodes $i$ and $j$, 
%% of the observations $\mathbf{c}_{ij}$ 
with $(i,j) \!\in\! \mathcal{C}$, is
\begin{align}
 & p(\mathbf{c}_{ij}|\bm{\theta}_{i},\bm{\theta}_{j};\Delta_{ij}) \label{eq:trueLH} \\ %\label{eq:LH_clock}
 & \quad=\, G_{ij} \exp\!\left(\rmv-\ist\frac{\left\Vert \mathbf{c}_{i\to j} \rmv-\bm{\psi}_{i\to j}\right\Vert ^{2}}{2\alpha_{j}^{2}\sigma_w^{2}} 
   - \frac{\left\Vert \mathbf{c}_{j\to i} \rmv-\bm{\psi}_{j\to i}\right\Vert ^{2}}{2\alpha_{i}^{2}\sigma_w^{2}} \right) ,\nonumber \\[-5.5mm]
\nonumber
\end{align}
where $G_{ij} \triangleq {(2\pi\alpha_{j}^{2}\sigma_w^{2})}^{-K_{ij}/2} {(2\pi\alpha_{i}^{2}\sigma_w^{2})}^{-K_{ji}/2}\rmv$,
$\bm{\psi}_{i\to j} \triangleq$\linebreak %%%%%
$\big[\psi^{(1)}_{i\to j}(\bm{\theta}_{i},\bm{\theta}_{j},\Delta_{ij}) \,\cdots\, \psi^{(K_{ij})}_{i\to j}(\bm{\theta}_{i},\bm{\theta}_{j},\Delta_{ij}) \big]^{\mathrm{T}}\!$,
\vspace{.4mm}
and $\bm{\psi}_{j\to i} \triangleq$\linebreak %%%%%
$\big[\psi^{(1)}_{j\to i}(\bm{\theta}_{j},\bm{\theta}_{i},\Delta_{ij}) \,\cdots\, \psi^{(K_{ji})}_{j\to i}(\bm{\theta}_{j},\bm{\theta}_{i},\Delta_{ij}) \big]^{\mathrm{T}}\!$.
\highlight{
Since the covariance depends on $\alpha_{i}$ and $\alpha_{j}$ and marginalization over these parameters is
%infeasible
not analytically tractable, the direct application of the likelihood function as in \cite{leng} for message passing is not straightforward. Moreover, the dependence of the unknown delay $\Delta_{ij}$ does not vanish in the presented distribution. 
In the following, we propose an approximation of \eqref{eq:trueLH} to circumvent these problems.\footnote{\highlight{In \cite{du13}, an alternative solution was proposed where the likelihood function is replaced by a surrogate function. The nonlinear dependencies and the delays were eliminated by scaling the measurement equations with the skew parameter, which destroys the likelihood property in the technical sense. 
%Through this modification it is not possible to formulate a likelihood function 
%and thus a Bayesian estimator.
%in the technical sense. 
%The modification does not preserve the likelihood property in the technical sense. 
Hence the estimator is not a MAP estimator.}
}}
% Observe that (\ref{eq:trueLH}) is difficult to work with, since (i)
% the factor 
% $G_{ij}$
% %outside the exponential 
% depends on $\alpha_{i}$ and $\alpha_{j}$; %and can be therefore not be marginalized over those values
% and (ii) the dependence on the unknown delay $\Delta_{ij}$. 
\highlight{Computing the Fischer information of \eqref{eq:trueLH} with respect to $\alpha_i$ (and similarly to $\alpha_j$), it can be seen that $G_{ij}$ has a smaller contribution\footnote{\highlight{The contribution to the Fisher information of the scaling factor $G_{ij}$ is $K_{ji}/{\alpha_i^2}$, and of the exponent $- \left(3 K_{ji} \sigma_w^2 + \|\mathbf{t}_{ij,0} \|^2 + \|\mathbf{t}_{ji,0} + \mathbf{1} \Delta_{ij} \|^2  \right)/({\alpha_i}^2 \sigma_w^2)$.}}
than the exponent,
%From a Fisher information analysis in $\alpha_i$ (and similarly in $\alpha_j$), it can be seen that $G_{ij}$ contributes less\footnote{The contribution to the Fisher information of the scaling factor $G_{ij}$ is $K_{ji}/{\alpha_i^2}$, and of the exponent $- \left(3 K_{ji} \sigma_w^2 + \|\mathbf{t}_{ij,0} \|^2 + \|\mathbf{t}_{ji,0} + \mathbf{1} \Delta_{ij} \|^2  \right)/({\alpha_i}^2 \sigma_w^2)$.}
%to the information regarding the skews, 
as long as $\|\mathbf{t}_{ij,0} \|^2 + \|\mathbf{t}_{ji,0} + \mathbf{1} \Delta_{ij} \|^2 \gg K_{ji} \sigma_w^2$, where $\mathbf{t}_{nm,0}$ is the collection of ${t}_{nm,0}^{(k)}, \, k \in \{1,\ldots, K_{nm}\}$. Thus, for practical scenarios we can approximate $G_{ij} \approx \tilde{G}_{ij} \triangleq {(2\pi \sigma_w^{2})}^{-K_{ij}/2} {(2\pi\sigma_w^{2})}^{-K_{ji}/2}\rmv$.
%The approximation introduces a bias in $\alpha_i$ of $\alpha_i^\ast \frac{\sigma_w^2 K_{ji} }{ \|\mathbf{t}_{ij,0} \|^2 + \|\mathbf{t}_{ji,0} + \mathbf{1} \Delta_{ij} \|^2 }$, where $\alpha_i^\ast$ is the true parameter. Depending on the distribution of the true transmission times and on the noise variance, this bias is kept small.
}
% The first problem can be avoided by approximating $\left(2\pi\alpha_{i}^{2}\sigma_{w}^{2}\right)^{-K_{ij}/2}$
% by $\left(2\pi\sigma_{w}^{2}\right)^{-K_{ij}/2}$, which
% leads a good approximation of (\ref{eq:trueLH}), as long as the clock
% skews are close to one, which is the physically relevant case. 
% We deal with the second problem 

\highlight{
The dependence of the unknown delay $\Delta_{ij}$ can be removed} by computing the maximum likelihood estimate
of $\Delta_{ij}$ and substituting the estimate back into the likelihood.
Taking the logarithm of (\ref{eq:trueLH}) and setting the derivative
with respect to $\Delta_{ij}$ to zero leads to the following estimate
\begin{equation}
\hat{\Delta}_{ij}(\bm{\theta}_{i}, \bm{\theta}_{j})=a_{i}\frac{1}{\alpha_{i}}+a_{j}\frac{1}{\alpha_{j}}+b_{ij}\frac{\beta_{i}}{\alpha_{i}}-b_{ij}\frac{\beta_{j}}{\alpha_{j}},\label{eq:delayEstimate}
\end{equation}
where $a_{i},a_{j},b_{ij}$ are functions of the observations, detailed
in Appendix \ref{ap:ML_dist_estimate}. Substituting (\ref{eq:delayEstimate})
in (\ref{eq:trueLH}) and considering the approximation of the normalization
constant leads to the following approximate likelihood function 
\begin{align}
\tilde{p} & (\mathbf{c}_{ij}|\bm{\vartheta}_i, \bm{\vartheta}_j)\propto\exp\left(-\frac{1}{2\sigma_{w}^{2}}\|\mathbf{A}_{ij}\bm{\vartheta}_{i}+\mathbf{B}_{ij}\bm{\vartheta}_{j}\|^{2}\right),\label{eq:apprx_LH_GaussForm_Hat}
\end{align}
with
\begin{align*}
  \mathbf{A}_{ij} & \triangleq \begin{bmatrix} 
      -\tilde{\mathbf{c}}_{i\to j}\!\! & \!\!\mathbf{1}_{K_{ij} }\\[.7mm]
      \mathbf{c}_{j\to i}\!\! & \!\!-\mathbf{1}_{K_{ji}} 
      \end{bmatrix}
      +[a_{i}\ +b_{ij}]\otimes\mathbf{1}_{K_{ij}+K_{ji},1},\\
  \mathbf{B}_{ij} &\triangleq \begin{bmatrix} 
      \mathbf{c}_{i\to j}\!\! & \!\!-\mathbf{1}_{K_{ij} }\\[.7mm]
      -\tilde{\mathbf{c}}_{j\to i}\!\! & \!\!\mathbf{1}_{K_{ji}} 
      \end{bmatrix} 
      +[a_{j}\ -b_{ij}]\otimes\mathbf{1}_{K_{ij}+K_{ji},1},
\end{align*}
 where $\otimes$ denotes the Kronecker product. Note that $\mathbf{A}_{ji} \neq \mathbf{B}_{ij}$, but $\mathbf{A}_{ji}^\mathrm{T}\mathbf{A}_{ji} = \mathbf{B}_{ij}^\mathrm{T}\mathbf{B}_{ij}$. The approximated
likelihood function in \eqref{eq:apprx_LH_GaussForm_Hat} no longer
contains the delay $\Delta_{ij}$ and can be interpreted as Gaussian
in the transformed parameters $\bm{\vartheta}_{i},\bm{\vartheta}_{j}$. As we will see
in Section \ref{sec:MP}, this latter observation has advantages in
the algorithm design for distributed parameter estimation since it
leads to simpler computation rules.

\subsection{Prior Distribution\label{sub:Prior-Information}}

Since our simplified likelihood function now has a Gaussian form in
the transformed clock parameters $\bm{\vartheta}_{i},\bm{\vartheta}_{j}$, 
we need to
select a suitable Gaussian prior so as to end up with a Gaussian posterior
distribution.

\highlight{As MNs induce a reference time in the WN, they have perfect knowledge of their clock parameters, modeled by
$p(\bm{\vartheta}_{i}) = \delta(\bm{\vartheta}_{i} - \bm{\vartheta}_i^\ast)$, $i \!\in\! \mathcal{M}$,
where  $\bm{\vartheta}_i^\ast$ denotes the true transformed clock parameter of MN $i$ and $\delta(\cdot)$ denotes the Dirac delta function.
}
\highlight{
For ANs, clock phases are in the most general case unbounded, and $\bb_i \triangleq \beta_{i}/\alpha_{i}$ can be modeled as 
having an as prior with infinite variance. For bounded intervals, a finite variance can be used.
%uninformative
%prior. 
The
clock skews depend on various random quantities such as environmental effects, production quality, and supply voltage. 
Moreover, the skews
of correctly working clocks are bounded in intervals close around 1, and we can use the approximation $\aa_i \triangleq 1/\alpha_{i} = 1/(1+\varepsilon_{i}) \approx 1-\varepsilon_{i}$ \cite{cristian89}, where $\varepsilon_i \triangleq \alpha_i - 1$. %, resulting in $\sigma_{\aa,i} = \sigma_{\alpha,i}$, which can be related to the specified clock accuracy. 
Finally, for the AN we use the Gaussian prior \cite{wu2}
$p(\bm{\vartheta}_{i})=\mathcal{N}(\bm{\mu}_{\text{p},i},\bm{\Sigma}_{\text{p},i})$, $i \!\in\! \mathcal{A}$,
with
$\bm{\mu}_{\text{p},i} = [ 1 \;\, 0]^{\mathrm{T}}\!$ (note that $\bm{\vartheta}_{i} = [ 1 \;\, 0]^{\mathrm{T}}$ would correspond to $\alpha_{i} \!=\! 1$ and 
$\beta_{i} \!=\! 0$) and $\bm{\Sigma}_{\text{p},i} = {\rm diag}\big\{ \sigma_{\rmv\aa_i}^{2}, \sigma_{\rmv\bb_i}^{2} \big\}$.
We set $\sigma_{\rmv\aa_i}^{2} \!=\rmv \sigma_{\rmv\alpha_i}^{2}$, where $\sigma_{\rmv\alpha_i}^{2}$ is %typically
related to the oscillator specification, and we choose $\sigma_{\rmv\bb_i}^{2}$ large, since 
limited prior information on the clock phase $\beta_{i}$ is available.
}

\subsection{Posterior Distribution and Estimator}

Putting together the approximate likelihood function from Section
\ref{sub:Likelihood-Function} with the prior in the transformed parameters
from Section \ref{sub:Prior-Information}, we find the following
posterior distribution in the transformed parameters 
\begin{align}
\tilde{p}(\bm{\vartheta}|\mathbf{c}) & \propto\prod_{i\in\mathcal{A}\cup\mathcal{M}}p(\bm{\vartheta}_{i})\prod_{(i,j)\in\mathcal{C}}\tilde{p}(\mathbf{c}_{ij}|\bm{\vartheta}_{i}, \bm{\vartheta}_{j}),\label{eq:posterior_fact_modified}
\end{align}
which is a Gaussian distribution in $\bm{\vartheta}$. The inverse covariance matrix
$\tilde{\bm{\Sigma}}^{-1}$ of this Gaussian turns out to be highly
structured, with block entries (for $i,j\in\mathcal{A}$) 
\begin{align}
\left[\tilde{\bm{\Sigma}}^{-1}\right]_{i,i} & =\bm{\Sigma}_{\text{p},i}^{-1}+\sum_{j\in\mathcal{T}_i}\frac{1}{\sigma^2_{w}}\mathbf{A}_{ij}^{\mathrm{T}}\mathbf{A}_{ij}^{} 
%+\frac{1}{\sigma^2_{w}}\mathbf{B}_{ji}^{\mathrm{T}}\mathbf{B}_{ji}^{}
\nonumber \\
\left[\tilde{\bm{\Sigma}}^{-1}\right]_{i,j} & =\left\{ \begin{matrix}\frac{1}{\sigma^2_{w}}\mathbf{A}_{ij}^{\mathrm{T}}\mathbf{B}_{ij}^{}
%+\frac{1}{\sigma^2_{w}}\mathbf{B}_{ji}^{\mathrm{T}}\mathbf{A}_{ji}^{} 
& \text{for }j\in\mathcal{T}_i\\
\mathbf{0} & \text{else.}
\end{matrix}\right.\label{eq:MAP_covar}
\end{align}
If we are able to marginalize $\tilde{p}(\bm{\vartheta}|\mathbf{c})$
to recover $\tilde{p}(\bm{\vartheta}_{i}|\mathbf{c})$, we can compute
the MAP estimate of $\bm{\vartheta}_{i}$, $i\in\mathcal{A}$
as
\begin{align}
\hat{\bm{\vartheta}}_{i} & =\arg\max_{\bm{\vartheta}_{i}}\tilde{p}(\bm{\vartheta}_{i}|\mathbf{c})\label{eq:MAP_est_modified}\\
 & =\arg\max_{\bm{\vartheta}_{i}}\int\,\tilde{p}(\bm{\vartheta}|\mathbf{c})\,\mathrm{d}\bm{\vartheta}_{\ist \bar{i}},\nonumber 
\end{align}
where $\bm{\vartheta}_{\ist \bar{i}}$ indicates that the integration is
over all $\bm{\vartheta}_{j}$ except $\bm{\vartheta}_{i}$. From $\hat{\bm{\vartheta}}_{i}$,
we can further determine the clock parameters by $\hat{\alpha}_{i}=1/[\hat{\bm{\vartheta}}_{i}]_{1}$
and $\hat{\beta}_{i}=[\hat{\bm{\vartheta}}_{i}]_{2}/[\hat{\bm{\vartheta}}_{i}]_{1}$,
where $[\cdot]_{m}$ extracts the $m$-th element of a vector. Solving
this problem in a distributed manner will be the topic of Section
\ref{sec:MP}. 

%--------------------------------------------------------------------
% SECTION: Bayesian Cramer Rao Bound
%--------------------------------------------------------------------

\section{Bayesian Cramér--Rao Bound}

\label{sec:BCRB} 
%--------------------------------------------------------------------
% SECTION: Bayesian Cramer Rao Bound
%--------------------------------------------------------------------
Based on the statistical models from Section \ref{sec:statmod}, it
is possible to derive fundamental performance bounds on the quality
of estimators. One such bound is the BCRB,
which gives a lower bound on the achievable estimation accuracy on
$\bm{\theta}_{i}$ \cite{trees}. The BCRB is derived based on the
Fisher information matrix, assuming known $\Delta_{ij}$ for every
link: 
\begin{align}
\mathbf{J} & =-\mathbb{E}_{{\boldsymbol{\theta},\mathbf{c}}}\begin{bmatrix}\left\{ \nabla_{{\boldsymbol{\theta}}}\left\{ \nabla_{{\boldsymbol{\theta}}}[\log p(\bm{\theta}|\mathbf{c};\boldsymbol{\Delta})]\right\} ^{\mathrm{T}}\right\} \end{bmatrix}\nonumber \\
 & =-\mathbb{E}_{{\boldsymbol{\theta},\mathbf{c}}}\begin{bmatrix}\left\{ \nabla_{{\boldsymbol{\theta}}}\left\{ \nabla_{{\boldsymbol{\theta}}}[\log p(\mathbf{c}|\bm{\theta};\boldsymbol{\Delta})]\right\} ^{\mathrm{T}}\right\} \end{bmatrix}\nonumber \\
 & \quad-\mathbb{E}_{{\boldsymbol{\theta}}}\begin{bmatrix}\left\{ \nabla_{{\boldsymbol{\theta}}}\left\{ \nabla_{{\boldsymbol{\theta}}}[\log p(\bm{\theta})]\right\} ^{\mathrm{T}}\right\} \end{bmatrix}\nonumber \\
 & =\mathbb{E}_{{\boldsymbol{\theta}}}[\mathbf{J}_{l}]+\mathbb{E}_{{\boldsymbol{\theta}}}[\mathbf{J}_{p}],\label{eq:fischer}
\end{align}
in which the matrix $\mathbf{J}_{p}$ represents the contribution
of the prior information, and is a diagonal matrix with block
entries equal to the covariances matrices of the priors.
The matrix $\mathbf{J}_{l}$ represents the contribution of the likelihood
function $p(\mathbf{c}|\bm{\theta};\Delta)=\mathcal{N}_{\mathbf{c}}(\bm{\mu}_{l},\bm{\Sigma}_{l})$, 
which is the product of the pairwise functions in (\ref{eq:trueLH}). It is computed as 
\begin{equation}
[\mathbf{J}_{l}]_{i,j}=\frac{\partial\boldsymbol{\mu}_{l}^{\mathrm{T}}}{\partial\boldsymbol{\theta}_{i}}\boldsymbol{\Sigma}_{l}^{-1}\frac{\partial\boldsymbol{\mu}_{l}}{\partial\boldsymbol{\theta}_{j}}+\frac{1}{2}\text{trace}\left[\boldsymbol{\Sigma}_{l}^{-1}\frac{\partial\boldsymbol{\Sigma}_{l}}{\partial\boldsymbol{\theta}_{i}}\boldsymbol{\Sigma}_{l}^{-1}\frac{\partial\boldsymbol{\Sigma}_{l}}{\partial\boldsymbol{\theta}_{j}}\right] \label{eq:FIM_LH}
\end{equation}
 for $i,j\in\mathcal{A}$. 
$\mathbf{J}_{l}$ has $2\times2$ non-zero blocks in the main diagonal
and in $i$-th row and $j$-th column when $j\in\mathcal{T}_i$. Thus,
it will have the same structure as the inverse covariance matrix in \eqref{eq:MAP_covar}.
Additional details are provided in Appendix \ref{ap:FIM}. Finally,
the BCRB on a certain parameter, say the $k$-th parameter in the
$2A$-dimensional vector $\bm{\theta}$, is given by 
\[
\text{BRCB}_{k}=\left[\mathbf{J}^{-1}\right]_{k,k}.
\]

%--------------------------------------------------------------------
% SECTION: Message Passing
%--------------------------------------------------------------------

\section{Distributed Parameter Estimation}

\label{sec:MP} 
%--------------------------------------------------------------------
% SECTION: Message Passing
%--------------------------------------------------------------------

To solve the marginalization in \eqref{eq:MAP_est_modified} in a
distributed way, we use approximate inference via message passing
on factor graphs. In the following, we describe the factor graph 
for the synchronization problem and motivate the use of message 
passing for optimum retrieval of posterior marginals. Finally, we 
derive two synchronization algorithms.

\subsection{Factor Graph}

\begin{figure}
  % priors
  \psfrag{pt01}[t][t][0.65]{\color[rgb]{0,0,0}\setlength{\tabcolsep}{0pt}\begin{tabular}{c}\raisebox{2.2mm}{$\, p(\bm{\vartheta}_1)$}\end{tabular}}
  \psfrag{pt02}[t][t][0.65]{\color[rgb]{0,0,0}\setlength{\tabcolsep}{0pt}\begin{tabular}{c}\raisebox{2.2mm}{$\, p(\bm{\vartheta}_2)$}\end{tabular}}
  \psfrag{pt03}[t][t][0.65]{\color[rgb]{0,0,0}\setlength{\tabcolsep}{0pt}\begin{tabular}{c}\raisebox{2.2mm}{$\, p(\bm{\vartheta}_3)$}\end{tabular}}
  \psfrag{pt4}[t][t][0.65]{\color[rgb]{0,0,0}\setlength{\tabcolsep}{0pt}\begin{tabular}{c}\raisebox{2.2mm}{$\, p(\bm{\vartheta}_4)$}\end{tabular}}
  \psfrag{pt5}[t][t][0.65]{\color[rgb]{0,0,0}\setlength{\tabcolsep}{0pt}\begin{tabular}{c}\raisebox{2.2mm}{$\, p(\bm{\vartheta}_5)$}\end{tabular}}
  % likelihoods
  \psfrag{pt12}[b][b][0.65]{\color[rgb]{0,0,0}\setlength{\tabcolsep}{0pt}\begin{tabular}{c}\vspace{-0.9mm}{$\, \, \tilde{p}(\mathbf{c}_{12} | \bm{\vartheta}_1, \bm{\vartheta}_2)$}\end{tabular}}
  \psfrag{pt23}[b][b][0.65]{\color[rgb]{0,0,0}\setlength{\tabcolsep}{0pt}\begin{tabular}{c}\vspace{-0.9mm}{$\, \, \tilde{p}(\mathbf{c}_{23} | \bm{\vartheta}_2, \bm{\vartheta}_3)$}\end{tabular}}
  \psfrag{pt24}[b][b][0.65]{\color[rgb]{0,0,0}\setlength{\tabcolsep}{0pt}\begin{tabular}{c}\vspace{-0.9mm}{$\, \, \tilde{p}(\mathbf{c}_{24} | \bm{\vartheta}_2, \bm{\vartheta}_4)$}\end{tabular}}
  \psfrag{pt25}[b][b][0.65]{\color[rgb]{0,0,0}\setlength{\tabcolsep}{0pt}\begin{tabular}{c}\vspace{-1.1mm}{$\hspace{2.5mm} \tilde{p}(\mathbf{c}_{25} | \bm{\vartheta}_2, \bm{\vartheta}_5)$}\end{tabular}}
  \psfrag{pt34}[b][b][0.65]{\color[rgb]{0,0,0}\setlength{\tabcolsep}{0pt}\begin{tabular}{c}\vspace{-0.9mm}{$\, \tilde{p}(\mathbf{c}_{34} | \bm{\vartheta}_3, \bm{\vartheta}_4)$}\end{tabular}}
  \psfrag{pt45}[b][b][0.65]{\color[rgb]{0,0,0}\setlength{\tabcolsep}{0pt}\begin{tabular}{c}\vspace{-0.9mm}{$\, \, \tilde{p}(\mathbf{c}_{45} | \bm{\vartheta}_4, \bm{\vartheta}_5)$}\end{tabular}}
  % variables
  \psfrag{t1}[b][b][0.65]{\color[rgb]{0,0,0}\setlength{\tabcolsep}{0pt}\begin{tabular}{c}\vspace{-1.4mm}{$\bm{\vartheta}_1 \ist$}\end{tabular}}
  \psfrag{t2}[b][b][0.65]{\color[rgb]{0,0,0}\setlength{\tabcolsep}{0pt}\begin{tabular}{c}\vspace{-1.4mm}{$\ist \bm{\vartheta}_2$}\end{tabular}}
  \psfrag{t3}[b][b][0.65]{\color[rgb]{0,0,0}\setlength{\tabcolsep}{0pt}\begin{tabular}{c}\vspace{-1.4mm}{$\bm{\vartheta}_3 \rmv$}\end{tabular}}
  \psfrag{t4}[b][b][0.65]{\color[rgb]{0,0,0}\setlength{\tabcolsep}{0pt}\begin{tabular}{c}\vspace{-1.4mm}{$\bm{\vartheta}_4 \rmv$}\end{tabular}}
  \psfrag{t5}[b][b][0.65]{\color[rgb]{0,0,0}\setlength{\tabcolsep}{0pt}\begin{tabular}{c}\vspace{-1.4mm}{$\bm{\vartheta}_5 \rmv$}\end{tabular}}
  %figure
  \centering{}\includegraphics[width=0.8\columnwidth]{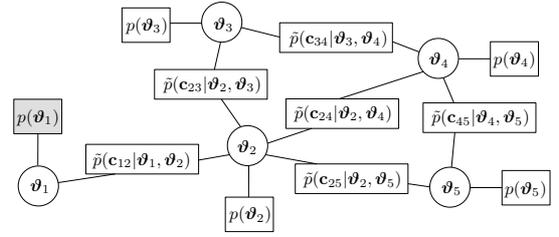}
  \caption{Factor graph of the posterior distribution for a 5 node network with $\mathcal{M}=\{1\}$ and $\mathcal{A}=\{2,3,4,5\}$.}
  \label{fig:FG_network}
\end{figure}

The factor graph associated to the factorization in \eqref{eq:posterior_fact_modified}
is found by drawing a variable vertex for every variable (drawn as
circles) and a factor/function vertex for every factor (drawn as rectangles).
Vertices are connected via edges according to their functional dependencies.
The factor graph%
\footnote{The representation differs slightly from the factor graph presented
in \cite{leng}, as in our case both nodes have access to the same
function vertex since they share the measurements. The presentation
in \cite{leng} accounts for 2 disjoint sets of measurements that
are not shared between the nodes \cite{wymeersch}.%
} that corresponds to the connectivity graph in Fig.~\ref{fig:connectivity}
is depicted in Fig.~\ref{fig:FG_network}. Note that every variable
vertex corresponds to the variables of a physical network node and
that every factor vertex corresponds to a measurement link in the
physical network. Thus, the structure of the connectivity graph is
kept in the factor graph: a tree connectivity remains as tree factor
graph, a star connectivity remains as star factor graph, and so on. 

Factor graphs are combined with message passing methods in order to
compute, e.g., marginal posteriors. Different message passing methods
lead to different performance/complexity trade-offs. A framework to
compare message passing method is found through variational free energy
minimization.

\subsection{Energy Minimization for Marginal Retrieval}

Our goal is to find practical methods to determine, exactly or approximately,
the marginals from \eqref{eq:MAP_est_modified}. From \cite{yedidia},
one strategy is to minimize the \textit{variational free energy} for
a positive function $b(\bm{\vartheta})$ approximating $\tilde{p}(\bm{\vartheta}|\mathbf{c})$:
\begin{equation}
b^{*}(\cdot)=\arg\min_{b(\cdot)} \int b(\bm{\vartheta})\,\log\frac{b(\bm{\vartheta})}{\tilde{p}(\bm{\vartheta}|\mathbf{c})}\,\mathrm{d}\bm{\vartheta}. \label{eq:FEmin}
\end{equation}

As algorithm designers, we can impose structure to the function
$b(\bm{\vartheta})$ to allow efficient solving of (\ref{eq:FEmin}).
We will consider two classes of functions: (i) the Bethe method,
in which $b(\bm{\vartheta})$ is constrained to be a product of factors
of the form $b_{i}(\bm{\vartheta}_{i})$ and $b_{ij}(\bm{\vartheta}_{i},\bm{\vartheta}_{j})$;
and (ii) the mean field method, which constrains $b(\bm{\vartheta})$
to be of the form $b(\bm{\vartheta})=\prod_{i}b_{i}(\bm{\vartheta}_{i})$.
Minimizing \eqref{eq:FEmin} subject to the constraints imposed
by the approximations, leads to the message passing rules \cite{yedidia}.
The message passing rules turn out to be the belief propagation (BP)
equations for class (i) and the mean field (MF) equations for class
(ii). 

In the following we use the shorthand ${p}_{ij}$ for $\tilde{p}(\mathbf{c}_{ij}|\bm{\vartheta}_{i}, \bm{\vartheta}_{j})$. %and $p_{i}$ for $p(\bm{\vartheta}_{i})$. 
Furthermore, since the approximated
joint posterior distribution in \eqref{eq:posterior_fact_modified} is
Gaussian in $\bm{\vartheta}$, we consider only messages that are Gaussian
in $\bm{\vartheta}$.

\subsection{Synchronization by Message Passing}

Above, we introduced two message passing schemes, BP and MF. By applying
both, we find two synchronization algorithms where network nodes cooperate by the exchange of messages. We now present the algorithms
in detail, and discuss their salient properties. A unified view of
the message passing is offered in Fig.~\ref{fig:msg_exchange}. 

% \begin{figure}
%   \centering{}
%   \includegraphics[scale=0.8]{fig/6_msg_exchange}
%   \caption{\label{fig:msg_exchange}Messages between a node pair $i,j$ of a
% 	general network. Since the measurements are shared, both nodes have
% 	access to the same function vertex.}
% \end{figure}

\begin{figure}[t]
	% Function vertices
	\psfrag{ptij}[b][b][0.8]{\color[rgb]{0,0,0}\setlength{\tabcolsep}{0pt}\begin{tabular}{c}\vspace{-0.9mm}{$\, \, \tilde{p}(\mathbf{c}_{ij} | \bm{\vartheta}_i, \bm{\vartheta}_j)$}\end{tabular}}
	% Variable vertices  
	\psfrag{ti}[b][b][0.8]{\color[rgb]{0,0,0}\setlength{\tabcolsep}{0pt}\begin{tabular}{c}{\vspace{-1.2mm} \hspace{-0.3mm}$\bm{\vartheta}_i \ist$}\end{tabular}}
	\psfrag{tj}[b][b][0.8]{\color[rgb]{0,0,0}\setlength{\tabcolsep}{0pt}\begin{tabular}{c}{\vspace{-0.7mm} \hspace{0.3mm}$ \bm{\vartheta}_j$}\end{tabular}}
% 	% Messages
	\psfrag{msg1j}[l][c][0.8]{\color[rgb]{0,0,0}\setlength{\tabcolsep}{0pt} \vspace{5mm}\hspace{-8.5mm}
	\begin{tabular}{lr}
	BP: $m_{\bm{\vartheta}_j \to p_{ij} } \hspace{5mm}$ & \eqref{eq:BPmsg_pij_in} \\
	MF: $b_{\bm{\vartheta}_j } \hspace{5mm}$ & \eqref{eq:MF_marg}
	\end{tabular}}
	\psfrag{msg1i}[l][c][0.8]{\color[rgb]{0,0,0}\setlength{\tabcolsep}{0pt} \vspace{7mm}\hspace{-8.5mm}
	\begin{tabular}{lr}
	BP: $m_{\bm{\vartheta}_i \to p_{ij} } \hspace{5mm}$ & \eqref{eq:BPmsg_pij_in} \\
	MF: $b_{\bm{\vartheta}_i } \hspace{5mm}$ & \eqref{eq:MF_marg}
	\end{tabular}}
	\psfrag{msg2j}[l][c][0.8]{\color[rgb]{0,0,0}\setlength{\tabcolsep}{0pt} \vspace{7mm}\hspace{-8mm}
	\begin{tabular}{lr}
	$m_{\bm{p_{ij} \to \vartheta}_j } \hspace{5mm}$ & \eqref{eq:BPmsg_thetai_in} or \eqref{eq:MFmsg_thetai_in} 
	\end{tabular}}
	\psfrag{msg2i}[l][c][0.8]{\color[rgb]{0,0,0}\setlength{\tabcolsep}{0pt} \vspace{7mm}\hspace{-8mm}
	\begin{tabular}{lr}
	$m_{\bm{p_{ij} \to \vartheta}_i } \hspace{5mm}$ & \eqref{eq:BPmsg_thetai_in} or \eqref{eq:MFmsg_thetai_in} 
	\end{tabular}}
	% Text
	\psfrag{ci}[l][c][0.8]{\color[rgb]{0.5,0.5,0.5}\setlength{\tabcolsep}{0pt}\begin{tabular}{c} \hspace{-2mm} \vspace{-1mm} Computed by node~$i$ \end{tabular}}
	\psfrag{cj}[l][c][0.8]{\color[rgb]{0.5,0.5,0.5}\setlength{\tabcolsep}{0pt}\begin{tabular}{c} \hspace{-31mm} \vspace{-1mm} Computed by node~$j$ \end{tabular}}
	% Figure
	\centering{}\includegraphics[width=0.7\columnwidth]{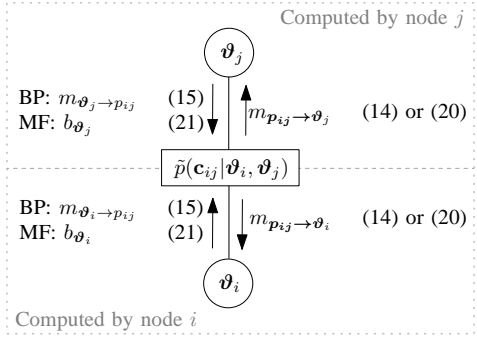}
	\caption{Messages between a node pair $i,j$ of a
	general network. Since the measurements are shared, both nodes have
	access to the same function vertex.}
	\label{fig:msg_exchange} 
\end{figure}	

\subsubsection{Belief Propagation}

The BP message from a factor vertex $p_{ij}$ to a variable vertex $\bm{\vartheta}_i$ is given by \cite[Eq. (6)]{kschischang}
\begin{align}
m_{{p}_{ij}\rightarrow\bm{\vartheta}_{i}}(\bm{\vartheta}_{i}) & =\int {p}_{ij}(\bm{\vartheta}_{i}, \bm{\vartheta}_{j})\ m_{\bm{\vartheta}_{j}\rightarrow{p}_{ij}}(\bm{\theta}_{j}') \,\mathrm{d}\bm{\theta}_j'\nonumber \\
 & \propto\mathcal{N}_{\bm{\vartheta}_{i}}\left(\boldsymbol{\mu}_{\text{in},ij},\boldsymbol{\Sigma}_{\text{in},ij}\right),\label{eq:BPmsg_thetai_in}
\end{align}
while the BP message from a variable vertex $\bm{\vartheta}_i$ to a factor vertex  $p_{ij}$ is given by \cite[Eq. (5)]{kschischang}
\begin{align}
m_{\bm{\vartheta}_{i}\rightarrow{p}_{ij}}(\bm{\vartheta}_{i}) & = p(\bm{\vartheta}_{i})\prod_{k\in\{\mathcal{T}_i\backslash j\}}\ m_{{p}_{ik}\rightarrow\bm{\vartheta}_{i}}(\bm{\vartheta}_{i})\nonumber \\
 & \propto\mathcal{N}_{\bm{\vartheta}_{i}}\left(\boldsymbol{\mu}_{\text{ext},ij},\boldsymbol{\Sigma}_{\text{ext},ij}\right),\label{eq:BPmsg_pij_in}
\end{align}
 where we use the index ``in'' for intrinsic and ``ext'' for extrinsic
with respect to a variable vertex. As depicted in Fig.~\ref{fig:msg_exchange}, 
each network node $i$ corresponding to the variable vertex $\bm{\vartheta}_i$ 
needs to compute its intrinsic and extrinsic message. Furthermore, note that for
BP, the extrinsic message $m_{\bm{\vartheta}_{i}\rightarrow{p}_{ij}}$
has to be determined separately for every node $j\in \mathcal{T}_i$. If the neighboring node is an
agent, $j\in\mathcal{T}_i\cap\mathcal{A}$,
the parameter updates (for detailed derivations, see Appendix \ref{ap:BP_update})
of \eqref{eq:BPmsg_thetai_in} are
\begin{subequations}
  \begin{align}
  \mathbf{Q} & =\mathbf{A}_{ij}^{\mathrm{T}}\mathbf{B}_{ij}^{}\left(\mathbf{B}_{ij}^{\mathrm{T}}\mathbf{B}_{ij}^{}+\sigma_{w}^{2}\boldsymbol{\Sigma}_{\text{ext},ji}^{-1}\right)^{-1}\nonumber \\
  \boldsymbol{\Sigma}_{\text{in},ij}^{-1} & =\frac{1}{\sigma_{w}^{2}}\mathbf{A}_{ij}^{\mathrm{T}}\mathbf{A}_{ij}^{}-\frac{1}{\sigma_{w}^{2}} \mathbf{Q}\,\mathbf{B}_{ij}^{\mathrm{T}}\mathbf{A}_{ij}^{}\label{eq:BP_Var_thetai_in}\\
  \boldsymbol{\Sigma}_{\text{in},ij}^{-1}\boldsymbol{\mu}_{\text{in},ij} & =-\mathbf{Q}\,\boldsymbol{\Sigma}_{\text{ext},ji}^{-1}\boldsymbol{\mu}_{\text{ext},ji},\label{eq:BP_Mean_thetai_in}
  \end{align}
\end{subequations}
 and if the neighboring node is a master, $j\in\mathcal{T}_i\cap\mathcal{M}$
\begin{subequations}
  \begin{align}
  \boldsymbol{\Sigma}_{\text{in},ij}^{-1} & =\frac{1}{\sigma_{w}^{2}}\mathbf{A}_{ij}^{\mathrm{T}}\mathbf{A}_{ij}^{}\label{eq:BP_Var_thetai_inMaster}\\
  \boldsymbol{\Sigma}_{\text{in},ij}^{-1}\boldsymbol{\mu}_{\text{in},ij}^{} & =-\frac{1}{\sigma_{w}^{2}}\mathbf{A}_{ij}^{\mathrm{T}}\mathbf{B}_{ij}^{}\boldsymbol{\mu}_{\text{ext},ji}^{}.\label{eq:BP_Mean_thetai_inMaster}
  \end{align}
\end{subequations}
 The parameter updates of \eqref{eq:BPmsg_pij_in} are 
\begin{subequations}
  \begin{align}
  \boldsymbol{\Sigma}_{\text{ext},ij}^{-1} & =\boldsymbol{\Sigma}_{\text{p},i}^{-1}+\sum_{k\in\{\mathcal{T}_i\backslash j\}}\boldsymbol{\Sigma}_{\text{in},ki}^{-1}\label{eq:BP_Var_pij_in}\\
  \boldsymbol{\Sigma}_{\text{ext},ij}^{-1}\boldsymbol{\mu}_{\text{ext},ij} & =\boldsymbol{\Sigma}_{\text{p},i}^{-1}\boldsymbol{\mu}_{\text{p},i}+\sum_{k\in\{\mathcal{T}_i\backslash j\}}\boldsymbol{\Sigma}_{\text{in},ki}^{-1}\boldsymbol{\mu}_{\text{in},ki}^{}.\label{eq:BP_Mean_pij_in}
  \end{align}
\end{subequations}
The approximate marginal is obtained by 
\begin{align}
b_{i}(\bm{\vartheta}_{i}) & \propto p(\bm{\vartheta}_{i}) \prod_{k\in\mathcal{T}_i}\ m_{{p}_{ik}\rightarrow\bm{\vartheta}_{i}}(\bm{\vartheta}_{i})\nonumber \\
 & \propto\mathcal{N}_{\bm{\vartheta}_{i}}\left(\boldsymbol{\mu}_{i},\boldsymbol{\Sigma}_{i}\right).\label{eq:BP_marg}
\end{align}
The parameters of the marginal belief \eqref{eq:BP_marg} are computed
from the parameters in \eqref{eq:BP_Var_pij_in} and \eqref{eq:BP_Mean_pij_in},
but with the additional summation over $j$.

In the communication between two connected nodes $i$ and $j$ as
in Fig.~\ref{fig:msg_exchange}, node $i$ transmits $m_{\bm{\vartheta}_{i}\rightarrow{p}_{ij}}$
to $j$ and vice versa. The receiving node then computes its intrinsic
message to the variable vertex (e.g., node $j$ computes $m_{{p}_{ij}\rightarrow\bm{\vartheta}_{j}}$).
As a node $i$ has evaluated the intrinsic messages from all its neighbors,
it can determine again its extrinsic messages. After $I$ iterations,
every node $i$ computes the marginal belief $b_{j}(\bm{\vartheta}_j)$ and thereof 
the MAP estimates of its clock parameters.

% \highlight{After a measurement phase of $K_{ij}$ packet transmissions, each node $i$ needs to
% compute $\forall j \in \mathcal{T}_i$ the factors $\mathbf{A}_{ij}^{\mathrm{T}}\mathbf{A}_{ij}$, $\mathbf{A}_{ij}^{\mathrm{T}}\mathbf{B}_{ij}$
% and $\mathbf{A}_{ij}^{\mathrm{T}}\mathbf{B}_{ij}$, involving the cost $C \approx 9 \, |\mathcal{T}_i| \,(K_{ij}+K_{ji}) \, O$. After that, the computational
% effort for \eqref{eq:BP_Mean_thetai_in} -- \eqref{eq:BP_Mean_pij_in} per iteration is $C \approx (5 \, | \mathcal{T}_i | ^2 + 35 \, | \mathcal{T}_i | \,) \, O$.
% }

\subsubsection{Mean Field}

The MF message from a factor vertex $p_{ij}$ to a variable vertex $\bm{\vartheta}_i$ is given by \cite[Eq. (14)]{dauwels}
\begin{align}
m_{{p}_{ij}\rightarrow\bm{\vartheta}_{i}}(\bm{\vartheta}_{i}) & =\exp\left(\int\log\left({p}_{ij}(\bm{\vartheta}_{i}, \bm{\vartheta}_{j})\right)\ b_{j}(\bm{\vartheta}_{j}) \, \mathrm{d}\bm{\theta}_{j}'\right)\nonumber \\
 & \propto\mathcal{N}_{\bm{\vartheta}_{i}}\left(\boldsymbol{\mu}_{\text{in},ij},\boldsymbol{\Sigma}_{\text{in},ij}\right), \label{eq:MFmsg_thetai_in}
\end{align}
and the message from a variable vertex $\bm{\vartheta}_i$ to a factor vertex $p_{ij}$ is given by the
belief \cite[Eq. (16)]{dauwels}
\begin{align}
b_{i}(\bm{\vartheta}_{i}) & \propto p(\bm{\vartheta}_{i})\prod_{k\in\mathcal{T}_i}\ m_{{p}_{ik}\rightarrow\bm{\vartheta}_{i}}(\bm{\vartheta}_{i})\nonumber \\
 & \propto\mathcal{N}_{\bm{\vartheta}_{i}}\left(\boldsymbol{\mu}_{i},\boldsymbol{\Sigma}_{i}\right).\label{eq:MF_marg}
\end{align}
% Note, that for MF the marginal belief $b_{i}(\bm{\theta}_i')$ needs to be passed
% instead of $m_{\bm{\vartheta}_{i}\rightarrow{p}_{ij}}(\bm{\theta}_i')$ at BP.
The corresponding parameter updates (for detailed derivations, see
Appendix \ref{ap:MF_update}) are
\begin{subequations}
  \begin{align}
  \boldsymbol{\Sigma}_{\text{in},ij}^{-1} & =\frac{1}{\sigma_{w}^{2}}\mathbf{A}_{ij}^{\mathrm{T}}\mathbf{A}_{ij}^{}\label{eq:MF_Var_thetai_in}\\
  \boldsymbol{\Sigma}_{\text{in},ij}^{-1}\boldsymbol{\mu}_{\text{in},ij} & =-\frac{1}{\sigma_{w}^{2}}\mathbf{A}_{ij}^{\mathrm{T}}\mathbf{B}_{ij}^{}\boldsymbol{\mu}_{j}^{}\label{eq:MF_Mean_thetai_in}
  \end{align}
\end{subequations}
and 
\begin{subequations}
  \begin{align}
  \boldsymbol{\Sigma}_{i}^{-1} & =\boldsymbol{\Sigma}_{\text{p},i}^{-1}+\sum_{k\in\mathcal{T}_i}\boldsymbol{\Sigma}_{\text{in},ki}^{-1}\label{eq:MF_Var_pij_in}\\
  \boldsymbol{\Sigma}_{i}^{-1}\boldsymbol{\mu}_{i} & =\boldsymbol{\Sigma}_{\text{p},i}^{-1}\boldsymbol{\mu}_{\text{p},i}+\sum_{k\in\mathcal{T}_i}\boldsymbol{\Sigma}_{\text{in},ki}^{-1}\boldsymbol{\mu}_{\text{in},ki}.\label{eq:MF_Mean_pij_in}
  \end{align}
\end{subequations}
 For MF, two connected nodes $(i,j)$ only need to exchange their
beliefs instead of extrinsic information (see Fig.~\ref{fig:msg_exchange}).
Since the same information is sent to all neighbors, this can also
be performed in a broadcast scheme. From the belief, the receiver
can then compute the intrinsic message \eqref{eq:MFmsg_thetai_in}.

% \highlight{
% After the measurement phase with $K_{ij}$ broadcasts, the factors $\mathbf{A}_{ij}^{\mathrm{T}}\mathbf{A}_{ij}$ and $\mathbf{A}_{ij}^{\mathrm{T}}\mathbf{B}_{ij}$ are computed with  $C \approx 6 \, |\mathcal{T}_i| \,(K_{ij}+K_{ji}) \, O$, which is less than for BP.
% Note also, that in every iteration only \eqref{eq:MF_Mean_thetai_in} needs to be updated. This results in the reduced
% computation cost of $C \approx 14 \, | \mathcal{T}_i |  \, O$ per iteration.
% }
\pagebreak
\subsection{Convergence}
\label{ssec:convergence}
\highlight{ \subsubsection{Mean Field}
%The advantage of the MF method is its guaranteed convergence \cite[pp. 466]{bishop}. Moreover, for Gaussian graphical models, MF
%converges to the true mean vectors \cite[pp.136]{wainwright}, as the factorization in node potentials always results to a convex problem.
MF optimizes node potentials, and for successive message updates, it is known to converge \cite[Theorem 11.10]{koller09} as the energy functional is monotonically decreasing and bounded. In Gaussian models, depending on the message ordering, MF converges to the true mean vectors \cite[pp.136]{wainwright}. %In the following numerical analysis we applied a modified parallel schedule, which differs from the guaranteed schedules, and always observed convergence.
\subsubsection{Belief Propagation}
BP optimizes node and edge potentials, and convergence in cyclic graphs depends on the underlying system. For Gaussian models,  
several sufficient conditions based on analysis of message propagation on the computation tree exist. These include 
diagonal dominance \cite{weiss} 
or walk-summability \cite{malioutov06}, and FG normalizabilitiy \cite{malioutov}, all of which can be evaluated via 
the information matrix \eqref{eq:MAP_covar}. 
In the following, we will prove the convergence of the proposed algorithms based on FG normalizabilitiy, 
which is a variant of the walk-sum interpretation
  %\footnote{It is not possible to adapt the convergence proof of \cite{leng}, as a special BP variant, also referred to as SPAWN \cite{wymeersch}, is used. %%Moreover, in contrast to the algorithms
  %% presented in this paper, the method from \cite{leng} does not guarantee convergence to
  %% the correct means or covariances.}.
%Moreover, the extension from scalar to the vector valued variables is not straightforward.}.
\begin{theorem}%[BP variances for informative prior without master nodes] 
\label{pr:Conv_prior}
  The variances of the proposed BP algorithm converge for connected networks without MNs, if each node has a prior with finite variance on skew and phase.
\end{theorem}
\begin{proof}
See Appendix \ref{sec:proof1}. 
\end{proof}
\begin{theorem}%[LBP variances for uninformative prior in the presence of master nodes]
      \label{pr:Conv_master}
      The variances of the proposed BP algorithm converge for all connected networks with at least one MN.
      %The variances of the proposed BP algorithm converge for connected networks with master nodes and for uninformative priors on clock parameters.
\end{theorem}
\begin{proof}
See Appendix \ref{sec:proof2}. 
\end{proof}}

\highlight{%The proofs of the theorems are found in appendix~\ref{ap:Conv_proofs}.
Once the variances converge, the mean updates follow a linear system. As shown in \cite{malioutov}, the convergence to the correct means can be forced by sufficient damping. In our numerical analysis, we did not encounter a single case where damping was necessary.
%By sufficient damping, they can be forced to converge to the correct means \cite[p.95]{malioutov}. Note that in the numerical analysis of this work we have not encountered a single case where this was necessary.
}

\subsection{Scheduling and Implementation Aspects}

\label{ssec:implemenationAspects}

In this section, we discuss message scheduling and ways to efficiently
combine timing information exchange and message passing in real applications.

%\subsubsection{Flooding schedule}
We consider a general topology as in Fig.~\ref{fig:FG_network}.
Every node $i\in\mathcal{A}$ runs the same algorithm, and computes
the message parameters to/from the function vertices as depicted in
Fig.~\ref{fig:msg_exchange}. Therefore, the node requires $m_{\bm{\vartheta}_{j}\rightarrow{p}_{ij}}(\bm{\vartheta}_{j})$
to compute $m_{{p}_{ij}\rightarrow\bm{\vartheta}_{i}}(\bm{\vartheta}_{i})$.
Together with the prior information, the node then computes the outgoing
message $m_{\bm{\vartheta}_{i}\rightarrow{p}_{ij}}(\bm{\vartheta}_{i})$,
which is sent to neighbor $j$ for the next iteration. In order to
start this procedure, all $m_{\bm{\vartheta}_{j}\rightarrow p_{ji}}(\bm{\vartheta}_{j})$
in the factor graph have to be initialized. This can be done by setting
them to uniform distributions, with zero mean value and infinite covariance. 

\highlight{
In general, all nodes work in parallel for all iterations. As discussed in Sec.~\ref{ssec:convergence}, MF convergence guarantees are only available for specific schedules. Since these are generally not practical in real applications, we propose a mixed serial/parallel MF schedule as follows:
% Note that the MF update rules in \eqref{ap:MF_update} do not depend on 
% the covariances of neighboring nodes. Thus, parameter uncertainties are not considered, 
% and the convergence speed can be significantly increased by a specific schedule as follows:
A node only 
%joins the protocol,
updates its beliefs
if information from a MN 
has propagated via any path to the node.}
Hence, in the first iteration only MNs $m\in\mathcal{M}$
propagate messages. In the second iteration, also their neighbors
$j\in \mathcal{T}_m$ will send messages, in the third their neighbors'
neighbors and so on. The schedule is serial in the initial information propagation, 
\highlight{and a compromise of the successive message updates and a parallel schedule. A similar schedule can be applied for MF if no MNs are available. Thereby, only nodes in the neighborhood of the AN which initializes 
the synchronization start to join the protocol. In this case, the initializing node also adjusts its clock parameters.} 

%\subsubsection{Piggybacking} 
Finally, our derivations were based on the assumption that measurements
were collected first, and then message passing was carried out. Since both
phases rely on the exchange of packets between nodes, it is possible
to combine them, thereby increasing the number of measurements
as message passing iterations progress. Such piggybacking is beneficial
in real applications. In order to successfully start the algorithm,
a minimum number of measurement packets has to be exchanged between
every node pair to provide initial timing information. This is due
to the required matrix inversions in the message parameter computations.
%\highlight{Piggybacking skipped}

\begin{table}[t]
\renewcommand{\arraystretch}{1.2} 
\caption{Complexity per estimation update}
\label{tab:alg_comp} 
\centering{}\vspace{-0.3cm}
\begin{tabular}{|l||c|c|c|c|c|}
\hline
& Num. Operations & Num. Transm.\\
\hline \hline
MF & once: $6 \, |\mathcal{T}_i| \,(K_{ij}+K_{ji}) \, O$ & $K_{ij}$\\
& $14 \, | \mathcal{T}_i |  \, O$ & 1\\
\hline
BP & once: $9 \, |\mathcal{T}_i| \,(K_{ij}+K_{ji}) \, O$ & $K_{ij}$ \\
 & $(5 \, | \mathcal{T}_i | ^2 + 35 \, | \mathcal{T}_i | \,) \, O$ & 1\\
\hline
ATS \cite{schenato} & $21 \, | \mathcal{T}_i | \, O$ & 1\\
\hline 
ADMM \cite{zennaro} & $9 I \,| \mathcal{T}_i |  O$ & $2 I$ \\%
\hline
LC \cite{solis} & $29 \, | \mathcal{T}_i | \, O$ & $2$\\
\hline
\end{tabular}
\vspace{-2mm}
\end{table}

\subsection{Comparison with State-of-the-Art Algorithms}
%
%\highlight{
%We have shown that our proposed fully distributed algorithms are 
%guaranteed to converge with and without MNs.
%Moreover, MF shows a significant complexity reduction compared to BP, which 
%is quadratic in the number of neighbors.
%}
We will now describe the main similarities and differences of BP and
MF to the previously described state-of-the-art algorithms from Section~\ref{sec:stateoftheart}.
We will use the shorthand: ATS is Average TimeSync from \cite{schenato},
ADMM is the method proposed in \cite{zennaro}, and LC is the loop constraint
method from \cite{solis}.
\highlight{
\subsubsection{Complexity} 
We compare the complexity per node $i \in \mathcal{A}$ with the number of operations needed per estimation update.
% using the estimates of the computational costs $C$
% for the basic mathematical operations per parameter update. 
%A approximate estimate of $C$ is given 
%at the end of each algorithm presentation.
In Table \ref{tab:alg_comp} we provide a complexity estimate, where 
the operations $+,-,*,\sqrt{\cdot},\log$ and $\exp$ are equated with by one operation cost $O$, and only factors containing the number of neighbors $| \mathcal{T}_i|$ are considered.
% The complexity scales with the number of neighbors $| \mathcal{T}_i|$, and the relevant term is depicted in Tab.~\ref{tab:alg_comp}.
% Also, the number of packet transmissions per node, needed for one estimation update, is depicted.
As MF and BP require a measurement phase, 
they must additionally transmit measurement packets and then compute matrix products of $\mathbf{A}_{ij}$ and $\mathbf{B}_{ij}$.
%an additional computation line for MF to compute $\mathbf{A}_{ij}^{\mathrm{T}}\mathbf{A}_{ij}$ and $\mathbf{A}_{ij}^{\mathrm{T}}\mathbf{B}_{ij}$, and for BP to additionally compute  $\mathbf{B}_{ij}^{\mathrm{T}}\mathbf{B}_{ij}$ is added. 
For ADMM, $I=1$ inner iterations and the estimation 
%algorithm of \cite{noh} with $K_n$ measurements were considered.
of $\alpha_{ij}$ using a PLL were considered\footnote{\highlight{The estimation accuracy was considered of $0.5\,$ppm, which corresponds to observations on Texas Instruments ez430-RF2500 evaluation boards. If no PLL information is available, $\alpha_{ij}$ can be achieved using \cite{noh} based on time stamped packet exchange and additional computations.}}.
It can be seen that all methods scale linearly with the number of neighbors $| \mathcal{T}_i|$, only BP scales quadratically. ADMM has the lowest complexity per transmitted packet, followed by LC, MF, ATS, and BP. For a full complexity analysis, the total number of transmitted packets upon convergence must be considered. This is done in Sec.~\ref{ssec:num_comp_methods}.}
% After having passed a measurement phase, MF shows the lowest complexity per update, followed by LC, ATS, ADMM and BP.
%The computations at the end of the measurement phase need to be further accounted for BP and MF.
\subsubsection{Delay sensitivity} 
Propagation delays $\delta_{ij}$ influence the performance of the algorithm if not considered correctly. 
BP, MF and LC algorithms consider the delay as symmetric and unknown. Moreover, in ADMM and in ATS it is disregarded and equated to zero.
Thus, the accuracy of ADMM and ATS is decreased if $\delta_{ij}$ increases, i.e., by additional deterministic delays.
\subsubsection{Master nodes} As discussed at the convergence section, BP and MF can operate with and without MNs. 
LC, ATS and ADMM do not consider the use of a time reference.
\subsubsection{Broadcast protocols} For ATS, LC, and MF, a node $i$ needs to pass identical values to all neighbors $j \in \mathcal{T}_i$. ADMM and BP have destination-specific messages. % from a $i$ to each $j \in \mathcal{T}_i$. 
In principle, this can also be accomplished by broadcast messages, when stacking the information to all neighbors in one packet. Thus all algorithms can be used with broadcast protocols, however ADMM and BP have higher bandwidth requirements.

The remaining question regarding the estimation accuracy is addressed in
the following section, where a superior behavior of BP and MF is observed.

%--------------------------------------------------------------------
% SECTION: Convergence
%--------------------------------------------------------------------

\section{Numerical Analysis}

\label{sec:NumAn} %--------------------------------------------------------------------
% SECTION: Numerical Analysis
%--------------------------------------------------------------------

\begin{figure}[t]
  % text strings:
  \psfrag{s01}[t][t][0.8]{\color[rgb]{0,0,0}\setlength{\tabcolsep}{0pt}\begin{tabular}{c}$x$-position [m]\end{tabular}}%
  \psfrag{s02}[b][b][0.8]{\color[rgb]{0,0,0}\setlength{\tabcolsep}{0pt}\begin{tabular}{c}$y$-position [m]\end{tabular}}%
  % xticklabels:
  \psfrag{x01}[t][t][0.7]{0}%
  \psfrag{x02}[t][t][0.7]{250}%
  \psfrag{x03}[t][t][0.7]{500}%
  \psfrag{x04}[t][t][0.7]{750}%
  \psfrag{x05}[t][t][0.7]{1000}%
  % yticklabels:
  \psfrag{v01}[r][r][0.7]{0}%
  \psfrag{v02}[r][r][0.7]{250}%
  \psfrag{v03}[r][r][0.7]{500}%
  \psfrag{v04}[r][r][0.7]{750}%
  \psfrag{v05}[r][r][0.7]{1000}%
  \centering{\includegraphics[width=0.7\columnwidth]{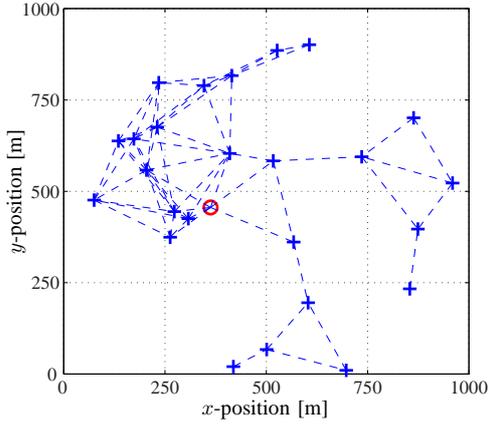}
  \caption{\label{fig:network}Randomly connected network with $M=1$
  (circle) and $A=25$ (cross).}
  }
\end{figure}

\subsection{Simulation Settings}

If not specified otherwise, we use the delay and noise setting from
the measurements in Fig.~\ref{fig:delay_spread}. In particular,
the noise standard deviation is $\sigma_{w}=93\,$ns and the
deterministic delay $\Delta_{ij}\rmv = \rmv T_{c}+T_{f,ij}$ comprises a computational
time $T_{c}\rmv=  \linebreak 7.6\,\mu\mathrm{s}$ and the flight time $T_{f,ij}\rmv=\rmv d_{ij}/v$,
where $d_{ij}$ is the distance between nodes $i$ and $j$, and $v$
is the speed of light.
%Simulations are carried out on 50 random networks where $M=1$ master nodes and $A=9$ agent nodes are uniformly placed in a $100\,$m$\times100\,$m square area. Each node has a communication radius of $50\,$m which
%determines the connectivity of the network. 
Simulations were carried out on randomly generated topologies with 26 nodes, as depicted in Fig.~\ref{fig:network}.
% An example topology is
% shown in Fig.~\ref{fig:network}. 
We further select $K_{ij}\rmv=\rmv K_{ji}$ between
all node pairs, where a measure-ment from node $i$ to node $j$ is
always followed with a measurement from node $j$ to node $i$. The
time between two subsequent measurements is set to $10\,$ms. The clock
skews are drawn from a normal distribution corresponding to a 100 ppm
specification,$\,\alpha_{i}\!\sim\!\mathcal{N}(\rmv 1,\rmv 10^{-8}\rmv)$, which represents the prior distribution.
The clock phases are drawn from a uniform distribution in the interval $[-10,+10]\,$s, 
where the Gaussian prior was specified with zero mean and a standard deviation of $\sigma_{\beta,i} = 5.8\,$s.
%\textcolor{red}{{[}what value do you assume in the algorithm for the
%variance on phase?{]} }
In the following, we will use the root mean square error as performance measure,
denoted by ``RMSE of phase'' and ``RMSE of skew''. Since not all 
competing algorithms rely on a MN, the RMSE to the true clock parameters 
is not a meaningful measure in a direct comparison. Thus, for algorithms not 
supporting MNs, the RMSE is evaluated with respect to the network's 
mean error of skew and phase.%, denoted by ``VAR of phase'' and ``VAR of skew''.

\subsection{Study of BP and MF Synchronization}

\subsubsection{Convergence Rate}

\begin{figure}[t]
  \begin{centering}
    \begin{psfrags}
    % text strings:
    \psfrag{s02}[b][b][\mlab]{\color[rgb]{0,0,0}\setlength{\tabcolsep}{0pt}\begin{tabular}{c}RMSE on phase [ms]\end{tabular}}%
    %\psfrag{s01}[b][b]{0}
    \psfrag{s05}[l][l]{\color[rgb]{0,0,0}BCRB}%
    \psfrag{s06}[l][l][\slab]{\color[rgb]{0,0,0}MF}%
    \psfrag{s07}[l][l][\slab]{\color[rgb]{0,0,0}BP}%
    \psfrag{s08}[l][l][\slab]{\color[rgb]{0,0,0}BCRB}%
    \psfrag{s10}[][]{\color[rgb]{0,0,0}\setlength{\tabcolsep}{0pt}\begin{tabular}{c} \end{tabular}}%
    \psfrag{s11}[][]{\color[rgb]{0,0,0}\setlength{\tabcolsep}{0pt}\begin{tabular}{c} \end{tabular}}%
    \psfrag{s12}[t][t][\mlab]{\color[rgb]{0,0,0}\setlength{\tabcolsep}{0pt}\begin{tabular}{c}Number of message iterations\end{tabular}}%
    \psfrag{s13}[b][b][\mlab]{\color[rgb]{0,0,0}\setlength{\tabcolsep}{0pt}\begin{tabular}{c}RMSE on skew [ppm]\end{tabular}}%
    %\psfrag{s16}[l][l]{\color[rgb]{0,0,0}BCRB}%
%     \psfrag{s17}[l][l][\slab]{\color[rgb]{0,0,0}MF}%
%     \psfrag{s18}[l][l][\slab]{\color[rgb]{0,0,0}BP}%
%     \psfrag{s19}[l][l][\slab]{\color[rgb]{0,0,0}BCRB}%
%     \psfrag{s21}[][]{\color[rgb]{0,0,0}\setlength{\tabcolsep}{0pt}\begin{tabular}{c} \end{tabular}}%
%     \psfrag{s22}[][]{\color[rgb]{0,0,0}\setlength{\tabcolsep}{0pt}\begin{tabular}{c} \end{tabular}}%
%     \psfrag{s47}[][]{\color[rgb]{0,0,0}\setlength{\tabcolsep}{0pt}\begin{tabular}{c} \end{tabular}}%
%     \psfrag{s48}[][]{\color[rgb]{0,0,0}\setlength{\tabcolsep}{0pt}\begin{tabular}{c} \end{tabular}}%
%     \psfrag{s36}[][]{\color[rgb]{0,0,0}\setlength{\tabcolsep}{0pt}\begin{tabular}{c} \end{tabular}}%
%     \psfrag{s37}[][]{\color[rgb]{0,0,0}\setlength{\tabcolsep}{0pt}\begin{tabular}{c} \end{tabular}}%
    % xticklabels:
	\psfrag{x01}[t][t]{0}%
	\psfrag{x02}[t][t]{0.1}%
	\psfrag{x03}[t][t]{0.2}%
	\psfrag{x04}[t][t]{0.3}%
	\psfrag{x05}[t][t]{0.4}%
	\psfrag{x06}[t][t]{0.5}%
	\psfrag{x07}[t][t]{0.6}%
	\psfrag{x08}[t][t]{0.7}%
	\psfrag{x09}[t][t]{0.8}%
	\psfrag{x10}[t][t]{0.9}%
	\psfrag{x11}[t][t]{1}%
    \psfrag{x28}[t][t][\slab]{0}%
    \psfrag{x29}[t][t][\slab]{1}%
    \psfrag{x30}[t][t][\slab]{2}%
    \psfrag{x31}[t][t][\slab]{3}%
    \psfrag{x32}[t][t][\slab]{4}%
    \psfrag{x33}[t][t][\slab]{5}%
    \psfrag{x34}[t][t][\slab]{6}%
    \psfrag{x35}[t][t][\slab]{7}%
    \psfrag{x36}[t][t][\slab]{8}%
    \psfrag{x37}[t][t][\slab]{9}%
    \psfrag{x15}[t][t][\slab]{7}%
    \psfrag{x16}[t][t][\slab]{8}%
    \psfrag{x17}[t][t][\slab]{9}%
    \psfrag{x18}[t][t][\slab]{0}%
    \psfrag{x19}[t][t][\slab]{1}%
    \psfrag{x20}[t][t][\slab]{2}%
    \psfrag{x21}[t][t][\slab]{3}%
    \psfrag{x22}[t][t][\slab]{4}%
    \psfrag{x23}[t][t][\slab]{5}%
    \psfrag{x24}[t][t][\slab]{6}%
    \psfrag{x25}[t][t][\slab]{7}%
    \psfrag{x26}[t][t][\slab]{8}%
    \psfrag{x27}[t][t][\slab]{9}%
    \psfrag{x12}[t][t][\slab]{7}%
    \psfrag{x13}[t][t][\slab]{8}%
    \psfrag{x14}[t][t][\slab]{9}%
    % yticklabels:
	\psfrag{v01}[r][r]{0}%
	\psfrag{v02}[r][r]{0.1}%
	\psfrag{v03}[r][r]{0.2}%
	\psfrag{v04}[r][r]{0.3}%
	\psfrag{v05}[r][r]{0.4}%
	\psfrag{v06}[r][r]{0.5}%
	\psfrag{v07}[r][r]{0.6}%
	\psfrag{v08}[r][r]{0.7}%
	\psfrag{v09}[r][r]{0.8}%
	\psfrag{v10}[r][r]{0.9}%
	\psfrag{v11}[r][r]{1}%
    \psfrag{v12}[r][r][\slab]{$0.05$}%
    \psfrag{v13}[r][r][\slab]{$0.1$}%
    \psfrag{v14}[r][r][\slab]{$0.15$}%
    \psfrag{v15}[r][r][\slab]{$0.2$}%
    \psfrag{v16}[r][r][\slab]{$0.25 \, 10^{-4}$}%
    \psfrag{v17}[r][r][\slab]{$0.5\, 10^{-4}$}%
    \psfrag{v18}[r][r][\slab]{$0.74\, 10^{-4}$}%
    \psfrag{v19}[r][r][\slab]{$10^{-4}$}%
    \psfrag{v20}[r][r][\slab]{$10^{-2}$}%
    \psfrag{v21}[r][r][\slab]{$1$}%
    \psfrag{v22}[r][r][\slab]{$10^{2}$}%
    \psfrag{v23}[r][r][\slab]{$10^{-5}$}%
    \psfrag{v24}[r][r][\slab]{$10^{-2}$}%
    \psfrag{v25}[r][r][\slab]{$10$}%
    \psfrag{v26}[r][r][\slab]{$10^{4}$}%
    \psfrag{ypower3}[Bl][Bl]{}%
    % Figure:
    \includegraphics[width=0.9\columnwidth]{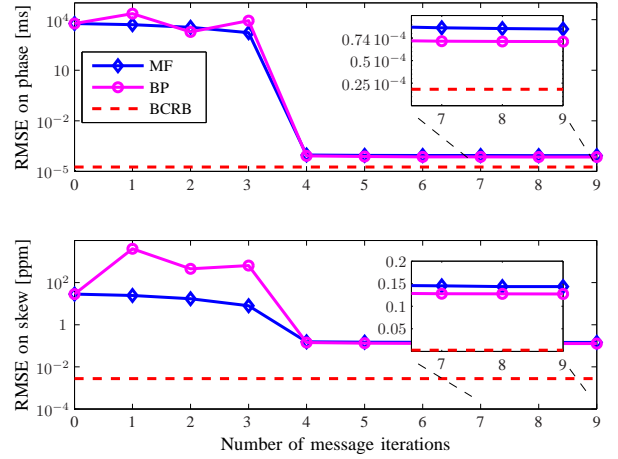}
    \caption{\label{fig:var_Iter}Convergence of parameter estimates with $K_{ij} = K_{ji} = 20$.}
    \end{psfrags}
  \end{centering}
  
\end{figure}

In Fig.~\ref{fig:var_Iter}, we show the BP and MF mean square error
of the phase and skew estimates as a function of the iteration index.
For the topology in Fig.~\ref{fig:network}, we observed that 
both algorithms converge after a number of iterations
that correspond to the largest multi-hop distance of a node to a master
node.%
\footnote{For the topology depicted in Fig.~\ref{fig:network}, which is part
of the randomly generated topologies, the largest multi-hop distance
is 4. This observation corresponds with the results shown by the simulation
in Fig.~\ref{fig:var_Iter}.%
} Furthermore, both algorithms converge to the same values. 
The gap
between estimation accuracy and BCRB %(averaged over all nodes and transformed to the correct dimension) 
arises due to 
% the estimation
% of the transformed parameters $\bm{\vartheta}$ instead of the original
% parameters $\bm{\theta}$.
the prior uncertainty of the clock phases.

\begin{figure}
\centering{
% \centering{}\includegraphics[width=1\columnwidth]{fig/7_study_Noise_MSE}
% \caption{\label{fig:var_sigma_time}Variation of measurement noise between
% every node pair with 40 measurements.}
  \begin{psfrags}
    \psfrag{s13}[b][b][\mlab]{\color[rgb]{0,0,0}\setlength{\tabcolsep}{0pt}\begin{tabular}{c}RMSE on skew [ppm]\end{tabular}}%
    \psfrag{s05}[l][l]{\color[rgb]{0,0,0}asdddd}%
    \psfrag{s06}[l][l][\slab]{\color[rgb]{0,0,0}MF - s}%
    \psfrag{s07}[l][l][\slab]{\color[rgb]{0,0,0}MF - p}%
    \psfrag{s08}[l][l][\slab]{\color[rgb]{0,0,0}BP}%
    \psfrag{s10}[][]{\color[rgb]{0,0,0}\setlength{\tabcolsep}{0pt}\begin{tabular}{c} \end{tabular}}%
    \psfrag{s11}[][]{\color[rgb]{0,0,0}\setlength{\tabcolsep}{0pt}\begin{tabular}{c} \end{tabular}}%
    \psfrag{s12}[t][t][\mlab]{\color[rgb]{0,0,0}\setlength{\tabcolsep}{0pt}\begin{tabular}{c}Number of message iterations\end{tabular}}%
    \psfrag{s02}[b][b][\mlab]{\color[rgb]{0,0,0}\setlength{\tabcolsep}{0pt}\begin{tabular}{c}RMSE on phase [ms]\end{tabular}}%
    %\psfrag{s16}[l][l]{\color[rgb]{0,0,0}asdddd}%
%     \psfrag{s17}[l][l][\slab]{\color[rgb]{0,0,0}MF - s}%
%     \psfrag{s18}[l][l][\slab]{\color[rgb]{0,0,0}MF - p}%
%     \psfrag{s19}[l][l][\slab]{\color[rgb]{0,0,0}BP}%
%     \psfrag{s21}[][]{\color[rgb]{0,0,0}\setlength{\tabcolsep}{0pt}\begin{tabular}{c} \end{tabular}}%
%     \psfrag{s22}[][]{\color[rgb]{0,0,0}\setlength{\tabcolsep}{0pt}\begin{tabular}{c} \end{tabular}}%
    %
    % xticklabels:
	\psfrag{x01}[t][t]{0}%
	\psfrag{x02}[t][t]{0.1}%
	\psfrag{x03}[t][t]{0.2}%
	\psfrag{x04}[t][t]{0.3}%
	\psfrag{x05}[t][t]{0.4}%
	\psfrag{x06}[t][t]{0.5}%
	\psfrag{x07}[t][t]{0.6}%
	\psfrag{x08}[t][t]{0.7}%
	\psfrag{x09}[t][t]{0.8}%
	\psfrag{x10}[t][t]{0.9}%
	\psfrag{x11}[t][t]{1}%
    \psfrag{x12}[t][t][\slab]{0}%
    \psfrag{x13}[t][t][\slab]{20}%
    \psfrag{x14}[t][t][\slab]{40}%
    \psfrag{x15}[t][t][\slab]{60}%
    \psfrag{x16}[t][t][\slab]{0}%
    \psfrag{x17}[t][t][\slab]{20}%
    \psfrag{x18}[t][t][\slab]{40}%
    \psfrag{x19}[t][t][\slab]{60}%
    \psfrag{x20}[t][t][\slab]{0}%
    \psfrag{x21}[t][t][\slab]{100}%
    \psfrag{x22}[t][t][\slab]{200}%
    \psfrag{x23}[t][t][\slab]{300}%
    \psfrag{x24}[t][t][\slab]{400}%
    \psfrag{x25}[t][t][\slab]{500}%
    \psfrag{x26}[t][t][\slab]{600}%
    \psfrag{x27}[t][t][\slab]{700}%
    \psfrag{x28}[t][t][\slab]{800}%
    \psfrag{x29}[t][t][\slab]{900}%
    \psfrag{x30}[t][t][\slab]{0}%
    \psfrag{x31}[t][t][\slab]{100}%
    \psfrag{x32}[t][t][\slab]{200}%
    \psfrag{x33}[t][t][\slab]{300}%
    \psfrag{x34}[t][t][\slab]{400}%
    \psfrag{x35}[t][t][\slab]{500}%
    \psfrag{x36}[t][t][\slab]{600}%
    \psfrag{x37}[t][t][\slab]{700}%
    \psfrag{x38}[t][t][\slab]{800}%
    \psfrag{x39}[t][t][\slab]{900}%
    %
    % yticklabels:
	\psfrag{v01}[r][r]{0}%
	\psfrag{v02}[r][r]{0.1}%
	\psfrag{v03}[r][r]{0.2}%
	\psfrag{v04}[r][r]{0.3}%
	\psfrag{v05}[r][r]{0.4}%
	\psfrag{v06}[r][r]{0.5}%
	\psfrag{v07}[r][r]{0.6}%
	\psfrag{v08}[r][r]{0.7}%
	\psfrag{v09}[r][r]{0.8}%
	\psfrag{v10}[r][r]{0.9}%
	\psfrag{v11}[r][r]{1}%
    \psfrag{v12}[t][t][\slab]{$10^{-4}\,$ }%
    \psfrag{v13}[t][t][\slab]{$10^{-2}\,$ }%
    \psfrag{v14}[t][t][\slab]{$1$}%
    \psfrag{v15}[t][t][\slab]{$10^{-2}\,$ }%
    \psfrag{v16}[t][t][\slab]{$1$}%
    \psfrag{v17}[t][t][\slab]{$10^{2}$}%
    \psfrag{v18}[t][t][\slab]{$10^{4}$}%
    \psfrag{v19}[t][t][\slab]{$10^{-4}\,$ }%
    \psfrag{v20}[t][t][\slab]{$10^{-2}\,$ }%
    \psfrag{v21}[t][t][\slab]{$1$}%
    \psfrag{v22}[t][t][\slab]{$10^{2}$}%
    \psfrag{v23}[t][t][\slab]{$10^{4}$}%
    \psfrag{v24}[t][t][\slab]{$10^{-2}$ }%
    \psfrag{v25}[t][t][\slab]{$1$}%
    \psfrag{v26}[t][t][\slab]{$10^{2}$}%
    \psfrag{v27}[t][t][\slab]{$10^{4}$}%
    %
    % Figure:
    \includegraphics[width=0.9\columnwidth]{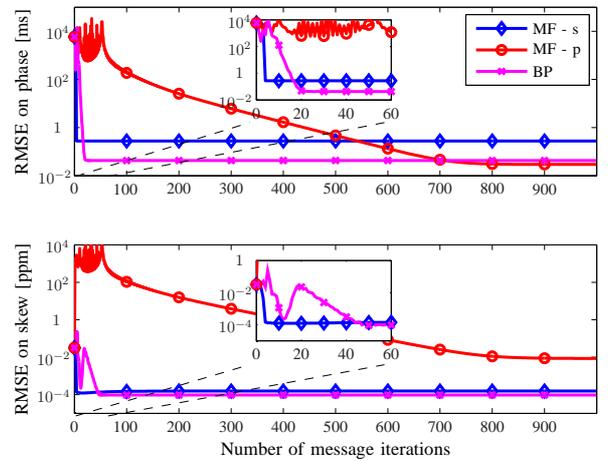}
    \caption{Convergence of clock parameters without MN, 
	every node pair with $K_{ij} = K_{ji} = 20$.}
    \label{fig:noMaster}
  \end{psfrags}%
  }
\end{figure}

\highlight{
As indicated by the proof of convergence, MF and BP do not require a MN for convergence if prior information with finite variances is available on all parameters. The convergence without MN is depicted in Fig.~\ref{fig:noMaster}, where BP uses a parallel schedule
%(``BP - p''), and MF a parallel (``MF - p'') and a schedule as described in Sec.~\ref{ssec:implemenationAspects} (``MF - s''). 
, and two schedules for MF are considered: ``MF - p'' is a parallel schedule, and ``MF - s'' is the serial schedule from Sec.~\ref{ssec:implemenationAspects}.
It can be observed that all algorithms converge, whereas ``MF - s'' has a significant higher convergence speed than ``MF - p''.
}
% 
% From the convergence analysis we can conclude, that the algorithms
% converge for practical relevant settings in a speed which can be predicted
% by the topology.

\subsubsection{Impact of Measurements}

\begin{figure}
\centering{
  \begin{psfrags}
    \psfrag{s01}[b][b][\mlab]{\color[rgb]{0,0,0}\setlength{\tabcolsep}{0pt}\begin{tabular}{c}RMSE on phase [ms]\end{tabular}}%
    \psfrag{s05}[l][l]{\color[rgb]{0,0,0}asddddd}%
    \psfrag{s06}[l][l][\slab]{\color[rgb]{0,0,0}MF}
    \psfrag{s07}[l][l][\slab]{\color[rgb]{0,0,0}BP}%
    \psfrag{s08}[l][l][\slab]{\color[rgb]{0,0,0}BCRB}%
    \psfrag{s10}[][]{\color[rgb]{0,0,0}\setlength{\tabcolsep}{0pt}\begin{tabular}{c} \end{tabular}}%
    \psfrag{s11}[][]{\color[rgb]{0,0,0}\setlength{\tabcolsep}{0pt}\begin{tabular}{c} \end{tabular}}%
    \psfrag{s12}[t][t][\mlab]{\color[rgb]{0,0,0}\setlength{\tabcolsep}{0pt}\begin{tabular}{c}Number of measurements\end{tabular}}%
    \psfrag{s13}[b][b][\mlab]{\color[rgb]{0,0,0}\setlength{\tabcolsep}{0pt}\begin{tabular}{c}RMSE on skew [ppm]\end{tabular}}%
    %\psfrag{s16}[l][l]{\color[rgb]{0,0,0}asddddd}%
%     \psfrag{s17}[l][l][\slab]{\color[rgb]{0,0,0}MF}%
%     \psfrag{s18}[l][l][\slab]{\color[rgb]{0,0,0}BP}%
%     \psfrag{s19}[l][l][\slab]{\color[rgb]{0,0,0}BCRB}%
%     \psfrag{s21}[][]{\color[rgb]{0,0,0}\setlength{\tabcolsep}{0pt}\begin{tabular}{c} \end{tabular}}%
%     \psfrag{s22}[][]{\color[rgb]{0,0,0}\setlength{\tabcolsep}{0pt}\begin{tabular}{c} \end{tabular}}%
    %
    % xticklabels:
    \psfrag{x01}[t][t][\slab]{$1$}%
    \psfrag{x02}[t][t][\slab]{$10$}%
    \psfrag{x03}[t][t][\slab]{$100$}%
    \psfrag{x04}[t][t][\slab]{$1000$}%
    \psfrag{x05}[t][t][\slab]{$1$}%
    \psfrag{x06}[t][t][\slab]{$10$}%
    \psfrag{x07}[t][t][\slab]{$100$}%
    \psfrag{x08}[t][t][\slab]{$1000$}%
    %
    % yticklabels:
    \psfrag{v07}[r][r][\slab]{$10^{-5}$}%
    \psfrag{v08}[r][r][\slab]{$10^{-4}$}%
    \psfrag{v09}[r][r][\slab]{$10^{-3}$}%
    \psfrag{v10}[r][r][\slab]{$10^{-2}$}%
    \psfrag{v01}[r][r][\slab]{$10^{-4}$}%
    \psfrag{v02}[r][r][\slab]{$10^{-3}$}%
    \psfrag{v03}[r][r][\slab]{$10^{-2}$}%
    \psfrag{v04}[r][r][\slab]{$0.1$}%
    \psfrag{v05}[r][r][\slab]{$1$}%
    \psfrag{v06}[r][r][\slab]{$10$}%
    %
    % Figure:
    \includegraphics[width=0.9\columnwidth]{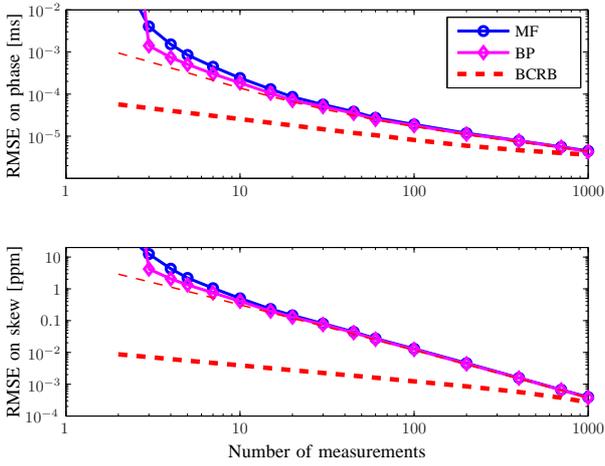}%
    \caption{RMSE and BCRB versus number of measurements, with BCRB for different phase distributions:
    bold with $[-10,10]\,$s, thin with $[-0.01,0.01]\,$s.}
    %The multiple BCRB lines are for different prior offset priors representing the interval $[-a,a]\,$s. The bold BCRB line is for $a = 0.01$, and $a \in [0.1,1,10]$ for the BCRB plots in decreasing order.}
    \label{fig:var_N_time}
  \end{psfrags}
}
\end{figure}

% Here, we vary the noise variance and the number of measurements, and
% provide results after $I=7$ message passing iterations. In both cases,
% we observed that BP and MF gave rise to very similar results. 

\begin{figure}
\centering{
% \includegraphics[width=1\columnwidth]{fig/7_study_Noise_MSE}
% \caption{\label{fig:var_sigma_time}Variation of measurement noise between
% every node pair with 40 measurements.}
  \begin{psfrags}
    \psfrag{s01}[b][b][\mlab]{\color[rgb]{0,0,0}\setlength{\tabcolsep}{0pt}\begin{tabular}{c}RMSE on phase [ms]\end{tabular}}%
    \psfrag{s05}[l][l][\slab]{\color[rgb]{0,0,0}asddddd}%
    \psfrag{s06}[l][l][\slab]{\color[rgb]{0,0,0}MF}%
    \psfrag{s07}[l][l][\slab]{\color[rgb]{0,0,0}BP}%
    \psfrag{s08}[l][l][\slab]{\color[rgb]{0,0,0}BCRB}%
    \psfrag{s10}[][]{\color[rgb]{0,0,0}\setlength{\tabcolsep}{0pt}\begin{tabular}{c} \end{tabular}}%
    \psfrag{s11}[][]{\color[rgb]{0,0,0}\setlength{\tabcolsep}{0pt}\begin{tabular}{c} \end{tabular}}%
    \psfrag{s12}[t][t][\mlab]{\color[rgb]{0,0,0}\setlength{\tabcolsep}{0pt}\begin{tabular}{c}Noise standard deviation [s]\end{tabular}}%
    \psfrag{s13}[b][b][\mlab]{\color[rgb]{0,0,0}\setlength{\tabcolsep}{0pt}\begin{tabular}{c}RMSE on skew [ppm]\end{tabular}}%
    %\psfrag{s16}[l][l]{\color[rgb]{0,0,0}asddddd}%
%     \psfrag{s17}[l][l][\slab]{\color[rgb]{0,0,0}MF}%
%     \psfrag{s18}[l][l][\slab]{\color[rgb]{0,0,0}BP}%
%     \psfrag{s19}[l][l][\slab]{\color[rgb]{0,0,0}BCRB}%
%     \psfrag{s21}[][]{\color[rgb]{0,0,0}\setlength{\tabcolsep}{0pt}\begin{tabular}{c} \end{tabular}}%
%     \psfrag{s22}[][]{\color[rgb]{0,0,0}\setlength{\tabcolsep}{0pt}\begin{tabular}{c} \end{tabular}}%
    %
    % xticklabels:
    \psfrag{x01}[t][t][\slab]{$10^{-10}$}%
    \psfrag{x02}[t][t][\slab]{$10^{-9}$}%
    \psfrag{x03}[t][t][\slab]{$10^{-8}$}%
    \psfrag{x04}[t][t][\slab]{$10^{-7}$}%
    \psfrag{x05}[t][t][\slab]{$10^{-6}$}%
    \psfrag{x06}[t][t][\slab]{$10^{-5}$}%
    \psfrag{x07}[t][t][\slab]{$10^{-4}$}%
    \psfrag{x08}[t][t][\slab]{$10^{-10}$}%
    \psfrag{x09}[t][t][\slab]{$10^{-9}$}%
    \psfrag{x10}[t][t][\slab]{$10^{-8}$}%
    \psfrag{x11}[t][t][\slab]{$10^{-7}$}%
    \psfrag{x12}[t][t][\slab]{$10^{-6}$}%
    \psfrag{x13}[t][t][\slab]{$10^{-5}$}%
    \psfrag{x14}[t][t][\slab]{$10^{-4}$}%
    %
    % yticklabels:
    \psfrag{v01}[r][r][\slab]{$10^{-4}$}%
    \psfrag{v02}[r][r][\slab]{$10^{-2}$}%
    \psfrag{v03}[r][r][\slab]{$1$}%
    \psfrag{v04}[r][r][\slab]{$10^{2}$}%
    \psfrag{v05}[r][r][\slab]{$10^{-8}$}%
    \psfrag{v06}[r][r][\slab]{$10^{-6}$}%
    \psfrag{v07}[r][r][\slab]{$10^{-4}$}%
    \psfrag{v08}[r][r][\slab]{$10^{-2}$}%
    %
    % Figure:
    \includegraphics[width=0.9\columnwidth]{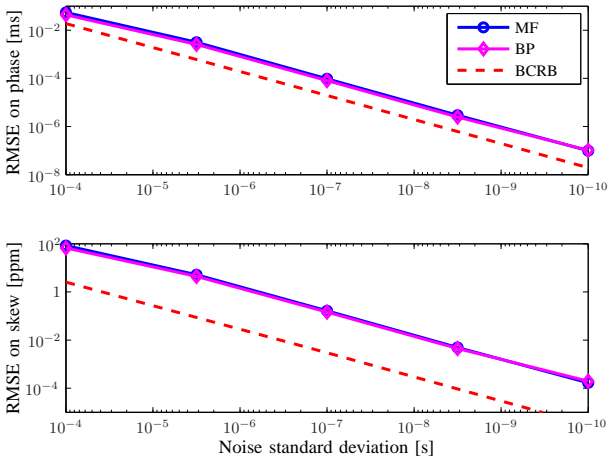}
    \caption{Variation of measurement noise between
	every node pair with $K_{ij} = K_{ji} = 20$.}
    \label{fig:var_sigma_time}
  \end{psfrags}%
  }
\end{figure}

The impact of the number of measurements after 7 message passing iterations is depicted in Fig.~\ref{fig:var_N_time},
where the estimation results are compared to the BCRB. With the number of measurements
the estimation accuracy increases and the gap to the BCRB is reduced. 
\highlight{Furthermore, we can observe that the prior uncertainty on the clock phases has a significant impact on the BCRB, but not on the MF or BP performance. The variation, which also explains the gap in Fig.~\ref{fig:var_Iter}, originates from the second order dependencies of the clock phases in the Fisher information matrix (see \eqref{eq:FIM_LH} and App.~\ref{ap:FIM}). It can be seen that for small phase intervals or for a large number of packets, the RMSE approaches the BCRB.}\\
Varying the noise variance $\sigma^2_{w}$ in Fig.~\ref{fig:var_sigma_time} reveals
its linear dependence to the estimation accuracy in double logarithmic
scale. Both plots can be used as design criteria for the synchronization
system.

\subsubsection{Scaling behavior}
\highlight{
To analyze the scaling behavior of the proposed methods, we used a grid network with equally spaced nodes in $x$ and $y$. Nodes are connected only to the next nodes in $x$ and $y$, and a MN was set to the one corner. In Fig.~\ref{fig:var_scale} it can be observed that the estimation accuracy decreases with increasing hop distance to the MN. This can be explained by the successive noise processes which are introduced in the connections.
}
\begin{figure}
  \centering
  \begin{psfrags}%
  % text strings:
  \psfrag{s01}[b][b][\mlab]{\color[rgb]{0,0,0}\setlength{\tabcolsep}{0pt}\begin{tabular}{c}RMSE on phase [ms]\end{tabular}}%
  \psfrag{s05}[l][l][\slab]{\color[rgb]{0,0,0}BCRB...}%
  \psfrag{s06}[l][l][\slab]{\color[rgb]{0,0,0}MF}%
  \psfrag{s07}[l][l][\slab]{\color[rgb]{0,0,0}BP}%
  \psfrag{s08}[l][l][\slab]{\color[rgb]{0,0,0}BCRB}%
  \psfrag{s10}[][][\slab]{\color[rgb]{0,0,0}\setlength{\tabcolsep}{0pt}\begin{tabular}{c} \end{tabular}}%
  \psfrag{s11}[][][\slab]{\color[rgb]{0,0,0}\setlength{\tabcolsep}{0pt}\begin{tabular}{c} \end{tabular}}%
  \psfrag{s12}[t][t][\slab]{\color[rgb]{0,0,0}\setlength{\tabcolsep}{0pt}\begin{tabular}{c}Number of nodes\end{tabular}}%
  \psfrag{s13}[b][b][\slab]{\color[rgb]{0,0,0}\setlength{\tabcolsep}{0pt}\begin{tabular}{c}RMSE on skew [ppm]\end{tabular}}%
%   \psfrag{s16}[l][l][\slab]{\color[rgb]{0,0,0}BCRB...}%
%   \psfrag{s17}[l][l][\slab]{\color[rgb]{0,0,0}MF}%
%   \psfrag{s18}[l][l][\slab]{\color[rgb]{0,0,0}BP}%
%   \psfrag{s19}[l][l][\slab]{\color[rgb]{0,0,0}BCRB}%
%   \psfrag{s21}[][][\slab]{\color[rgb]{0,0,0}\setlength{\tabcolsep}{0pt}\begin{tabular}{c} \end{tabular}}%
%   \psfrag{s22}[][][\slab]{\color[rgb]{0,0,0}\setlength{\tabcolsep}{0pt}\begin{tabular}{c} \end{tabular}}%
  % xticklabels:
  \psfrag{x01}[t][t][\slab]{0}%
  \psfrag{x02}[t][t][\slab]{20}%
  \psfrag{x03}[t][t][\slab]{40}%
  \psfrag{x04}[t][t][\slab]{60}%
  \psfrag{x05}[t][t][\slab]{80}%
  \psfrag{x06}[t][t][\slab]{100}%
  \psfrag{x07}[t][t][\slab]{120}%
  \psfrag{x08}[t][t][\slab]{140}%
  \psfrag{x09}[t][t][\slab]{160}%
  \psfrag{x10}[t][t][\slab]{180}%
  \psfrag{x11}[t][t][\slab]{200}%
  \psfrag{x12}[t][t][\slab]{0}%
  \psfrag{x13}[t][t][\slab]{20}%
  \psfrag{x14}[t][t][\slab]{40}%
  \psfrag{x15}[t][t][\slab]{60}%
  \psfrag{x16}[t][t][\slab]{80}%
  \psfrag{x17}[t][t][\slab]{100}%
  \psfrag{x18}[t][t][\slab]{120}%
  \psfrag{x19}[t][t][\slab]{140}%
  \psfrag{x20}[t][t][\slab]{160}%
  \psfrag{x21}[t][t][\slab]{180}%
  \psfrag{x22}[t][t][\slab]{200}%
  % yticklabels:
  \psfrag{v01}[r][r][\slab]{$10^{-3}$}%
  \psfrag{v02}[r][r][\slab]{$10^{-2}$}%
  \psfrag{v03}[r][r][\slab]{$10^{-1}$}%
  \psfrag{v04}[r][r][\slab]{$10^{0}$}%
  \psfrag{v05}[r][r][\slab]{$10^{-5}$}%
  \psfrag{v06}[r][r][\slab]{$10^{-4}$}%
  \psfrag{v07}[r][r][\slab]{$10^{-3}$}%
  %
  % Figure:
  \includegraphics[width=0.9\columnwidth]{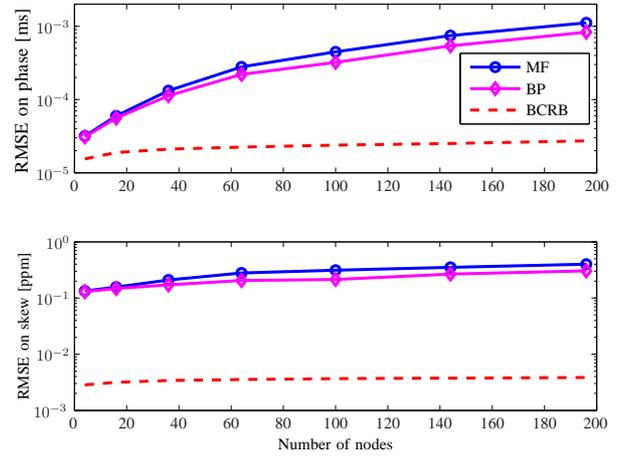}
    \caption{Algorithm performance after convergence in grid networks with increasing number of nodes, using $K_{ij} = K_{ji} = 20$.}
    \label{fig:var_scale}
  \end{psfrags}%
\end{figure}

\subsection{Comparison to other Algorithms}
\label{ssec:num_comp_methods}
We now compare BP and MF to other fully distributed state-of-the-art
algorithms\footnote{The following algorithm parameters are selected for \cite{schenato}: filter values $\rho_\eta = \rho_\alpha = \rho_o = 0.6$; for \cite{zennaro}: step size $\epsilon_\text{opt}$ according to \cite[Eq. (12)]{zennaro}, relative skew estimation as in \cite{noh} with $K_{ij} = K_{ji} = 5$; and for \cite{solis}: filter parameter $\lambda = 0.9$.}
 from Section \ref{sec:stateoftheart}. 
%Since in those methods
%time information and synchronization information exchange is interleaved,
%we apply the concept of piggybacking, introduced in Section \ref{ssec:implemenationAspects}.
%Furthermore, f
For a fair comparison, we set the computational delay
$T_{c}=0\,$s, since not all methods account for deterministic delays
between the nodes. Thus, the delay reduces to the time of flight, 
which is in the order of tenths of microseconds. 

\begin{figure}
\centering{
  \begin{psfrags}
    % text strings:
    \psfrag{s02}[b][b][\mlab]{\color[rgb]{0,0,0}\setlength{\tabcolsep}{0pt}\begin{tabular}{c}RMSE on Phase [ms]\end{tabular}}%
    \psfrag{s05}[l][l][\slab]{\color[rgb]{0,0,0}LC [18]}%
    \psfrag{s06}[l][l][\slab]{\color[rgb]{0,0,0}MF}%
    \psfrag{s07}[l][l][\slab]{\color[rgb]{0,0,0}BP}%
    \psfrag{s08}[l][l][\slab]{\color[rgb]{0,0,0}ATS \cite{schenato}}%
    \psfrag{s09}[l][l][\slab]{\color[rgb]{0,0,0}ADMM \cite{zennaro}}%
    \psfrag{s10}[l][l][\slab]{\color[rgb]{0,0,0}LC \cite{solis}}%
    \psfrag{s12}[][]{\color[rgb]{0,0,0}\setlength{\tabcolsep}{0pt}\begin{tabular}{c} \end{tabular}}%
    \psfrag{s13}[][]{\color[rgb]{0,0,0}\setlength{\tabcolsep}{0pt}\begin{tabular}{c} \end{tabular}}%
    \psfrag{s14}[t][t][\mlab]{\color[rgb]{0,0,0}\setlength{\tabcolsep}{0pt}\begin{tabular}{c}Number of broadcasts\end{tabular}}%
    \psfrag{s15}[b][b][\mlab]{\color[rgb]{0,0,0}\setlength{\tabcolsep}{0pt}\begin{tabular}{c}RMSE on skew [ppm]\end{tabular}}%
%     \psfrag{s18}[l][l][\slab]{\color[rgb]{0,0,0}LC [18]}%
%     \psfrag{s19}[l][l][\slab]{\color[rgb]{0,0,0}MF}%
%     \psfrag{s20}[l][l][\slab]{\color[rgb]{0,0,0}BP}%
%     \psfrag{s21}[l][l][\slab]{\color[rgb]{0,0,0}ATS \cite{schenato}}%
%     \psfrag{s22}[l][l][\slab]{\color[rgb]{0,0,0}ADMM \cite{zennaro}}%
%     \psfrag{s23}[l][l][\slab]{\color[rgb]{0,0,0}LC \cite{solis}}%
%     \psfrag{s25}[][]{\color[rgb]{0,0,0}\setlength{\tabcolsep}{0pt}\begin{tabular}{c} \end{tabular}}%
%     \psfrag{s26}[][]{\color[rgb]{0,0,0}\setlength{\tabcolsep}{0pt}\begin{tabular}{c} \end{tabular}}%
    %
    % xticklabels:
    \psfrag{x01}[t][t][\slab]{0}%
    \psfrag{x02}[t][t][\slab]{100}%
    \psfrag{x03}[t][t][\slab]{200}%
    \psfrag{x04}[t][t][\slab]{300}%
    \psfrag{x05}[t][t][\slab]{400}%
    \psfrag{x06}[t][t][\slab]{500}%
    \psfrag{x07}[t][t][\slab]{600}%
    \psfrag{x08}[t][t][\slab]{0}%
    \psfrag{x09}[t][t][\slab]{100}%
    \psfrag{x10}[t][t][\slab]{200}%
    \psfrag{x11}[t][t][\slab]{300}%
    \psfrag{x12}[t][t][\slab]{400}%
    \psfrag{x13}[t][t][\slab]{500}%
    \psfrag{x14}[t][t][\slab]{600}%
    %
    % yticklabels:
    \psfrag{v01}[r][r][\slab]{$10^{2}$}%
    \psfrag{v02}[r][r][\slab]{$10^{-2}$}%
    \psfrag{v03}[r][r][\slab]{$1$}%
    \psfrag{v04}[r][r][\slab]{$10^{-5}$}%
    \psfrag{v05}[r][r][\slab]{$10^{-2}$}%
    \psfrag{v06}[r][r][\slab]{$10$}%
    \psfrag{v07}[r][r][\slab]{$10^{4}$}%
        \includegraphics[width=0.9\columnwidth]{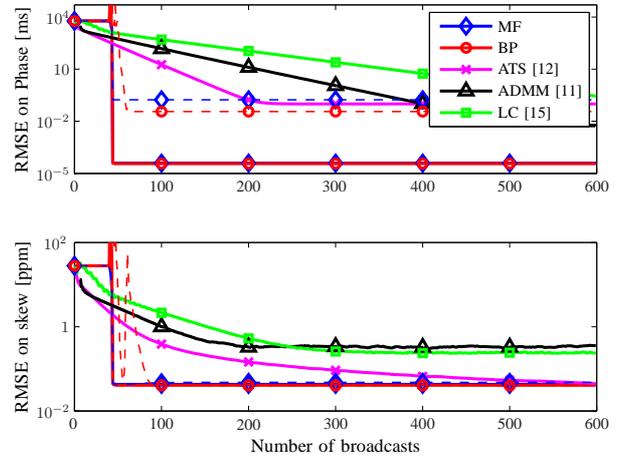}
	\caption{Variance of clock parameters using selected synchronization algorithms.}
	\label{fig:compare_methods}
  \end{psfrags}
  }
\end{figure}

In Fig.~\ref{fig:compare_methods},
%\ref{fig:var_meas_noDelay} and Fig.~\ref{fig:var_meas_Delay},
simulation results for phase and skew estimation are shown. The
simulations were performed on 
%the fixed network topology 
randomly created topologies as depicted
in Fig.~\ref{fig:network} and the results are averaged over 100 runs.
\highlight{The message passing algorithms use 40 measurements. The proposed MF and BP algorithms are evaluated with MN (solid line) and without MN (dashed line, with  decreased phase estimation accuracy). Using more measurements, the accuracy can be increased whereas more packet broadcasts are required. For the given setting, it can be observed that the proposed algorithms converge after around 100 message broadcasts, which is significantly lower than that of the competing methods. Moreover, the estimation accuracy of the network with MN is superior to those methods.}
%For both phase and skew estimation, the proposed algorithms outperform
%the competing ones for a fixed number of exchanged packets. 

\begin{figure}
\centering{
  \begin{psfrags}%
%     %\psfragscanon%
%     % text strings:
    \psfrag{s01}[t][t][\mlab]{\color[rgb]{0,0,0}\setlength{\tabcolsep}{0pt}\begin{tabular}{c}Number of broadcasts\end{tabular}}%
    \psfrag{s02}[b][b][\mlab]{\color[rgb]{0,0,0}\setlength{\tabcolsep}{0pt}\begin{tabular}{c}Computations\end{tabular}}%
    \psfrag{s05}[l][l][\slab]{\color[rgb]{0,0,0}LC [18]}%
    \psfrag{s06}[l][l][\slab]{\color[rgb]{0,0,0}MF}%
    \psfrag{s07}[l][l][\slab]{\color[rgb]{0,0,0}BP}%
    \psfrag{s08}[l][l][\slab]{\color[rgb]{0,0,0}ATS \cite{schenato}}%
    \psfrag{s09}[l][l][\slab]{\color[rgb]{0,0,0}ADMM \cite{zennaro}}%
    \psfrag{s10}[l][l][\slab]{\color[rgb]{0,0,0}LC \cite{solis}}%
    \psfrag{s12}[][][\slab]{\color[rgb]{0,0,0}\setlength{\tabcolsep}{0pt}\begin{tabular}{c} \end{tabular}}%
    \psfrag{s13}[][][\slab]{\color[rgb]{0,0,0}\setlength{\tabcolsep}{0pt}\begin{tabular}{c} \end{tabular}}%
    %
    % xticklabels:
	  \psfrag{x01}[t][t]{0}%
	  \psfrag{x02}[t][t]{0.1}%
	  \psfrag{x03}[t][t]{0.2}%
	  \psfrag{x04}[t][t]{0.3}%
	  \psfrag{x05}[t][t]{0.4}%
	  \psfrag{x06}[t][t]{0.5}%
	  \psfrag{x07}[t][t]{0.6}%
	  \psfrag{x08}[t][t]{0.7}%
	  \psfrag{x09}[t][t]{0.8}%
	  \psfrag{x10}[t][t]{0.9}%
	  \psfrag{x11}[t][t]{1}%
    \psfrag{x12}[t][t][\slab]{0}%
    \psfrag{x13}[t][t][\slab]{100}%
    \psfrag{x14}[t][t][\slab]{200}%
    \psfrag{x15}[t][t][\slab]{300}%
    \psfrag{x16}[t][t][\slab]{400}%
    \psfrag{x17}[t][t][\slab]{500}%
    \psfrag{x18}[t][t][\slab]{600}%
    %
    % yticklabels:
	  \psfrag{v01}[r][r]{0}%
	  \psfrag{v02}[r][r]{0.2}%
	  \psfrag{v03}[r][r]{0.4}%
	  \psfrag{v04}[r][r]{0.6}%
	  \psfrag{v05}[r][r]{0.8}%
	  \psfrag{v06}[r][r]{1}%
    \psfrag{v07}[r][r][\slab]{0}%
    \psfrag{v08}[r][r][\slab]{1}%
    \psfrag{v09}[r][r][\slab]{2}%
    \psfrag{v10}[r][r][\slab]{3}%
    \psfrag{v11}[r][r][\slab]{4}%
    \psfrag{v12}[r][r][\slab]{5}%
    \psfrag{ypower3}[Bl][Bl][\slab]{$\times 10^{4}$}%
    %
    % Figure:
    \includegraphics[width=0.9\columnwidth]{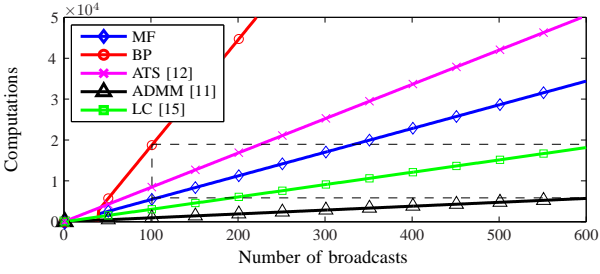}%
    \caption{Number of computations per node for selected synchronization algorithms.}
	\label{fig:compare_methods_comp}
    \end{psfrags}%
    }
\end{figure}

\highlight{In Fig.~\ref{fig:compare_methods_comp} the number of computations vs. broadcasts is depicted for a single node $i \in \mathcal{A}$ with $|\mathcal{T}_i| = 4$. 
We compare the complexity of the algorithms after convergence. MF and BP converge after 100 broadcasts (worst case for BP without MN), for ADMM after 400 broadcasts, and for LC and ATS after 600 iterations. 
It can be seen that ADMM has the lowest complexity, followed by MF, LC, BP, and finally ATS. However, ADMM uses PLL estimates and for optimized convergence, a centrally computed topology dependent step size. If no PLL estimates are accessible, estimates obtained with \cite{noh} would increase complexity and decrease convergence speed due to additional packet exchanges, which would rank ADMM after LC.
%As the proposed algorithms converge after 100 broadcasts, we compare the algorithms complexity at this position. The BP method has shows similar complexity as the LC, whereas the MF outperforms all competing algorithms. 
Thus, the simulation results indicate that MF has superior convergence speed and increased estimation accuracy while having lower computational requirements.}
% Despite
% having higher complexity, we can conclude that the proposed algorithms
% have superior estimation performance and convergence speed compared
% to the other methods.

%--------------------------------------------------------------------
% SECTION: CONCLUSION
%--------------------------------------------------------------------

\section{Conclusions}

\label{sec:Conclusion} In this paper, we presented two cooperative and fully distributed
network synchronization algorithms, which can be utilized when the
measurement noise is (approximately) Gaussian. Using standard
communication hardware, this approximation was verified by measurements.
The synchronization algorithm design is based on message passing in
a factor graph representation of the statistical model. Belief propagation
(BP) and mean field (MF) message passing were applied to perform MAP
estimation of the local clock parameters. We studied convergence,
convergence rate, and accuracy, and found that in all three criteria,
BP and MF are able to outperform existing algorithms. 
\highlight{Moreover, the MF method has significant advantages in computational complexity.
%The convergence study revealed that the methods 
Both BP and MF 
can perform synchronization 
with and without a global time reference.}

% However, our
% results are in some cases still far away from a newly derived Bayesian
% CRB, indicating the possibility for further improvement.

%\section*{Acknowledgment}

%--------------------------------------------------------------------
% SECTION: Appendix
%--------------------------------------------------------------------

\appendix

\subsection{ML-estimate of the delay}

\label{ap:ML_dist_estimate} The ML estimate of $\Delta_{ij}$ is
given by 
\[
\hat{\Delta}_{ij}=\arg\max_{\Delta_{ij}}\log p(\mathbf{c}_{ij}|\bm{\theta}_{i}, \bm{\theta}_{j};\Delta_{ij}),
\]
where 
\begin{align*}
 & \log p(\mathbf{c}_{ij}|\bm{\theta}_{i}, \bm{\theta}_{j};\Delta_{ij})\propto \\
 & \quad - \frac{\left\Vert \mathbf{c}_{i\to j}-\bm{\psi}_{i\to j}\right\Vert ^{2}}{2\alpha_{i}^{2}\sigma_{w}^{2}}
-\frac{\left\Vert \mathbf{c}_{j\to i}-\bm{\psi}_{j\to i})\right\Vert ^{2}}{2\alpha_{j}^{2}\sigma_{w}^{2}}.
\end{align*}
Since $\bm{\psi}_{i\to j}$ is linear
in $\Delta_{ij}$, taking the derivative of $\log p(\mathbf{c}_{ij}|\bm{\theta}_{i}, \bm{\theta}_{j};\Delta_{ij})$
with respect to $\Delta_{ij}$ and equating the result to zero,
immediately yields

\begin{align}
\hat{\Delta}_{ij}\left(\bm{\vartheta}_{i}, \bm{\vartheta}_{j}\right)= & \underbrace{\frac{K_{ij}\ \bar{c}_{i,ij}-K_{ji}\ \bar{c}_{i,ji}}{K_{ij}+K_{ji}}}_{a_{i}}\frac{1}{\alpha_{i}}\nonumber \\
 & +\underbrace{\frac{-K_{ij}\ \bar{c}_{j,ij}+K_{ji}\ \bar{c}_{j,ji}}{K_{ij}+K_{ji}}}_{a_{j}}\frac{1}{\alpha_{j}}\nonumber \\
 & +\underbrace{\frac{K_{ji}-K_{ij}}{K_{ij}+K_{ji}}}_{b_{ij}}\frac{\beta_{i}}{\alpha_{i}}+\underbrace{\frac{K_{ij}-K_{ji}}{K_{ij}+K_{ji}}}_{-b_{ij}}\frac{\beta_{j}}{\alpha_{j}},\label{eq:dist_est-1}
\end{align}
 with the averaged time stamps $\bar{c}_{i,ij}=1/K_{ij}\sum_{k}c_{i}(t_{ij,0}^{(k)})$,
$\bar{c}_{i,ji}=1/K_{ji}\sum c_{i}(t_{ji,1}^{(l)})$ of $i$, and
$\bar{c}_{j,ji}=1/K_{ji}\sum_{l}c_{j}(t_{ji,0}^{(l)})$, $\bar{c}_{j,ij}=1/K_{ij}\sum_{k}c_{j}(t_{ij,1}^{(k)})$
of $j$.

\subsection{Computation of Fischer Information Matrix}
\label{ap:FIM}

\subsubsection{Computation of $\mathbf{J}_{l}$}

From the true likelihood (\ref{eq:trueLH}) we have the parameter
set $\boldsymbol{\Sigma}_{l}^{-1},\boldsymbol{\mu}_{l}$ for every
connected node pair $(i,j)$ as 
\begin{align*}
\boldsymbol{\mu}_{l,ij} & =\begin{bmatrix} \bm{\psi}_{i\to j} \\ \bm{\psi}_{j\to i}\end{bmatrix}, &
\boldsymbol{\Sigma}_{l,ij} & =\begin{bmatrix}\alpha_{j}^{2}\sigma_{w}^{2}\mathbf{I}_{K_{ij}, K_{ij}} & \mathbf{0}\\
\mathbf{0} & \alpha_{i}^{2}\sigma_{w}^{2}\mathbf{I}_{K_{ji}, K_{ji}}
\end{bmatrix},
\end{align*}
where $\mathbf{I}_{k,l}$ denotes a $k\times l$ identity matrix. 
Applying \eqref{eq:FIM_LH} leads to the symmetric main diagonal
block entries 
\begin{align*}
  [\mathbf{J}_{l}]_{i,i}=& \,\frac{1}{\alpha_i^2 \sigma_{w}^{2}}\sum_{j\in\mathcal{T}_i} \left( \sum_{k=1}^{K_{ij}}\begin{bmatrix}(\tau_{ij,0}^{(k)})^{2} & \tau_{ij,0}^{(k)} \\
  \tau_{ij,0}^{(k)} & 1 
\end{bmatrix} \right.\\
 & \left. +\sum_{l=1}^{K_{ji}} \begin{bmatrix}(\tau_{ji,1}^{(l)})^{2} - 2\sigma_{w}^{2} & \tau_{ji,1}^{(l)} \\
 \tau_{ji,1}^{(l)} & 1
\end{bmatrix} \right).
\end{align*}

% \begin{align*}
%   [\mathbf{J}_{l}]_{i,i}=&\frac{1}{\sigma_{w}^{2}}\sum_{j\in\mathcal{T}_i} \left( \sum_{k=1}^{K_{ij}}\begin{bmatrix}\frac{(\phi_{ij,1}([\bm{\theta}_{i}]_{2})^{2}}{[\bm{\theta}_{i}]_{1}^{4}} & \frac{\phi_{ij,1}([\bm{\theta}_{i}]_{2})}{[\bm{\theta}_{i}]_{1}^{3}} \\
%   \frac{\phi_{ij,1}([\bm{\theta}_{i}]_{2})}{[\bm{\theta}_{i}]_{1}^{3}} & \frac{1}{[\bm{\theta}_{i}]_{1}^{2}} 
% \end{bmatrix} \right.\\
%  & \left. +\sum_{l=1}^{K_{ji}} \begin{bmatrix}\frac{\phi_{ji,2}(\bm{\theta}_{j})^{2} - 2\sigma_{w}^{2}}{[\bm{\theta}_{i}]_{1}^{2}} & \frac{\phi_{ji,2}([\bm{\theta}_{j}]_{2})}{[\bm{\theta}_{i}]_{1}^{2}} \\
% \frac{\phi_{ji,2}(\bm{\theta}_{j})}{[\bm{\theta}_{i}]_{1}^{2}} & \frac{1}{[\bm{\theta}_{i}]_{1}^{2}}
% \end{bmatrix} \right),
% \end{align*}
where $\tau_{ij,0}^{(k)} = (c_{i}(t_{ij,0}^{(k)})-\beta_i) / \alpha_i$ and $\tau_{ij,1}^{(k)} = ( c_{i}(t_{ij,0}^{(l)})- \beta_i)/ \alpha_i + \Delta_{ij}$ for any pair $(i,j) \in \mathcal{C}$. 
The off-diagonal block entries are 

\begin{align*}
 [\mathbf{J}_{l}]_{i,j}=&  -\frac{1}{\alpha_i \alpha_j \sigma_{w}^{2}}
  \left( \sum_{k=1}^{K_{ij}}\hspace{-0.1cm}\begin{bmatrix} \tau_{ij,0}^{(k)} \tau_{ij,1}^{(k)} & \tau_{ij,0}^{(k)}\\
 \tau_{ij,1}^{(k)} & 1
\end{bmatrix}+ \right. \\
 & \left. \sum_{l=1}^{K_{ji}}\hspace{-0.1cm}\begin{bmatrix} \tau_{ji,0}^{(l)} \tau_{ji,1}^{(l)} & \tau_{ji,0}^{(l)} \\
 \tau_{ji,1}^{(l)} & 1
\end{bmatrix}^{\mathrm{T}} \right)
\end{align*}

% \begin{align*}
%  [\mathbf{J}_{l}]_{i,j}=& -\frac{1}{\sigma_{w}^{2}}
%   \left( \sum_{k=1}^{K_{ij}}\hspace{-0.1cm}\begin{bmatrix}\frac{\phi_{ij,1}([\bm{\theta}_{i}]_{2})\phi_{ij,2}(\bm{\theta}_{i})}{[\bm{\theta}_{i}]_{1}^{2}[\bm{\theta}_{j}]_{1}} & \frac{\phi_{ij,1}([\bm{\theta}_{i}]_{2})}{[\bm{\theta}_{i}]_{1}^{2}[\bm{\theta}_{j}]_{1}}\\
% \frac{\phi_{ij,2}(\bm{\theta}_{i})}{[\bm{\theta}_{i}]_{1}[\bm{\theta}_{j}]_{1}} & \frac{1}{[\bm{\theta}_{i}]_{1}[\bm{\theta}_{j}]_{1}}
% \end{bmatrix}+ \right. \\
%  & \left. \sum_{l=1}^{K_{ji}}\hspace{-0.1cm}\begin{bmatrix}\frac{\phi_{ji,1}([\bm{\theta}_{j}]_{2}) \phi_{ji,2}(\bm{\theta}_{j})}{[\bm{\theta}_{i}]_{1}[\bm{\theta}_{j}]_{1}^{2}} & \frac{\phi_{ji,1}([\bm{\theta}_{j}]_{2})}{[\bm{\theta}_{i}]_{1}[\bm{\theta}_{j}]_{1}^{2}}\\
% \frac{\phi_{ji,2}(\bm{\theta}_{j})}{[\bm{\theta}_{i}]_{1} [\bm{\theta}_{j}]_{1}} & \frac{1}{[\bm{\theta}_{i}]_{1}[\bm{\theta}_{j}]_{1}}
% \end{bmatrix}^{\mathrm{T}} \right)
% \end{align*}
 for $j\in\mathcal{T}_i$ and $[\mathbf{J}_{l}]_{i,j}=\mathbf{0}$ else.

\subsubsection{Expectations of $\mathbb{E}_{{\boldsymbol{\theta}}}[\mathbf{J}_{l}]$
and $\mathbb{E}_{{\boldsymbol{\theta}}}[\mathbf{J}_{p}]$}

The expectations $\mathbb{E}_{{\boldsymbol{\theta}}}[\mathbf{J}_{l}]$
and $\mathbb{E}_{{\boldsymbol{\theta}}}[\mathbf{J}_{p}]$ have to
be taken over the inverse clock skews, i.e. $\mathbb{E}_{{\boldsymbol{\theta}}}[1/\alpha_{i}^{n}]$,
for $n$ up to 4. Since the clock skews are Gaussian distributed
and close to one, we use the approximations 
\begin{align*}
 & \mathbb{E}\left\{ {1}/{\alpha_{i}}\right\} \approx2-\mu_{\alpha,i}\qquad\mathbb{E}\left\{ {1}/{\alpha_{i}^{2}}\right\} \approx\sigma_{\alpha,i}^{2}+\mathbb{E}\left\{ {1}/{\alpha_{i}}\right\} ^{2}\\
 & \mathbb{E}\left\{ {1}/{\alpha_{i}^{3}}\right\} \approx\mathbb{E}\left\{ {1}/{\alpha_{i}}\right\} ^{3}+3\;\mathbb{E}\left\{ {1}/{\alpha_{i}}\right\} \sigma_{\alpha,i}^{2}\\
 & \mathbb{E}\left\{ {1}/{\alpha_{i}^{4}}\right\} \approx\mathbb{E}\left\{ {1}/{\alpha_{i}}\right\} ^{4}+6\;\mathbb{E}\left\{ {1}/{\alpha_{i}}\right\} ^{2}\sigma_{\alpha,i}^{2}+3\;\sigma_{\alpha,i}^{4}.
\end{align*}

\subsection{Belief Propagation Update Rules}

\label{ap:BP_update} Derivation of the message parameters for $m_{{p}_{ij}\rightarrow\bm{\vartheta}_{i}}(\bm{\theta}_{i})$
in \eqref{eq:BPmsg_thetai_in}: 
\begin{align*}
 & m_{{p}_{ij}\rightarrow\bm{\vartheta}_{i}}(\bm{\vartheta}_{i})=\int{p}_{ij}(\bm{\vartheta}_{i}, \bm{\vartheta}_{j})\ m_{\bm{\vartheta}_{j}\rightarrow{p}_{ij}}(\bm{\vartheta}_{j}) \, \mathrm{d}\bm{\theta}_{j}'\\
 & \ \propto\int\limits \exp\left(-\frac{1}{2\sigma_{w}^{2}}\|\mathbf{A}_{ij}\bm{\vartheta}_{i}+\mathbf{B}_{ij}\bm{\vartheta}_{j}\|^{2}\right) \times \\
  & \qquad \mathcal{N}_{\bm{\vartheta}_{j}}\left(\boldsymbol{\mu}_{\text{ext},ji},\boldsymbol{\Sigma}_{\text{ext},ji}\right) \, \mathrm{d}\bm{\theta}_{j}'\\
 & \ =\exp\left(\frac{G_{i}(\bm{\vartheta}_{i})}{2}\right)\int\exp\left(\frac{G_{ij}(\bm{\vartheta}_{i},\bm{\vartheta}_{j})}{2}\right)\, \mathrm{d}\bm{\theta}_{j}' \\
 & \ \propto\exp\left(\frac{G_{i}(\bm{\theta_{i}})}{2}\right)\\
 & \ \propto\mathcal{N}_{\bm{\vartheta}_{i}}\left(\boldsymbol{\mu}_{\text{in},ij},\boldsymbol{\Sigma}_{\text{in},ij}\right).
\end{align*}
 The functions $G_{i}(\bm{\vartheta}_{i})$ and $G_{ij}(\bm{\vartheta}_{i},\bm{\vartheta}_{j})$
are given by the exponent 
\begin{align*}
 & -\frac{1}{\sigma_{w}^{2}}\Big({\bm{\vartheta}_{j}}^{\mathrm{T}}\mathbf{B}_{ij}^{\mathrm{T}}\mathbf{B}_{ij}\bm{\vartheta}_{j}+2\,{\bm{\vartheta}_{i}}^{\mathrm{T}}\mathbf{A}_{ij}^{\mathrm{T}}\mathbf{B}_{ij}\bm{\vartheta}_{j}+{\bm{\vartheta}_{i}}^{\mathrm{T}}\mathbf{A}_{ij}^{\mathrm{T}}\mathbf{A}_{ij}\bm{\vartheta}_{i}\Big)\\
 & \quad-\Big(\bm{\vartheta}_{j}-\boldsymbol{\mu}_{\text{ext},ji}\Big)^{\mathrm{T}}\boldsymbol{\Sigma}_{\text{ext},ji}^{-1}\Big(\bm{\vartheta}_{j}-\boldsymbol{\mu}_{\text{ext},ji}\Big)\\
 & =G_{i}(\bm{\vartheta}_{i})+G_{ij}(\bm{\vartheta}_{i},\bm{\vartheta}_{j})
\end{align*}
 with 
\begin{align*}
&G_{ij} (\bm{\vartheta}_{i},\bm{\vartheta}_{j}) =-\Big(\bm{\vartheta}_{j}-\boldsymbol{\mu}'\Big)^{\mathrm{T}}{\boldsymbol{\Sigma}'}^{-1}\Big(\bm{\vartheta}_{j}-\boldsymbol{\mu}'\Big)\\
  & \qquad-\boldsymbol{\mu}_{\text{ext},ji}^{\mathrm{T}}\boldsymbol{\Sigma}_{\text{ext},ji}^{-1}\boldsymbol{\mu}_{\text{ext},ji}\\
&G_{i} (\bm{\vartheta}_{i}) =-\frac{1}{\sigma_{w}^{2}}{\bm{\vartheta}_{i}}^{\mathrm{T}}\Big(\underbrace{\mathbf{A}_{ij}^{\mathrm{T}}\mathbf{A}_{ij}-\mathbf{A}_{ij}^{T}\mathbf{B}_{ij}\frac{1}{\sigma_w^{2}}{\boldsymbol{\Sigma}'}\,\mathbf{B}_{ij}^{\mathrm{T}}\mathbf{A}_{ij}}_{\boldsymbol{\Sigma}_{\text{in},ij}^{-1}}\Big){\bm{\vartheta}_{i}}\\
 & \qquad -\frac{2}{\sigma_{w}^{2}}{\bm{\vartheta}_{i}}^{\mathrm{T}}\underbrace{\mathbf{A}_{ij}^{\mathrm{T}}\mathbf{B}_{ij} {\boldsymbol{\Sigma}'}\,\boldsymbol{\Sigma}_{\text{ext},ji}^{-1}\boldsymbol{\mu}_{\text{ext},ji}}_{\boldsymbol{\Sigma}_{\text{in},ij}^{-1}\boldsymbol{\mu}_{\text{ext},ji}}
\end{align*}
 and 
\begin{align*}
{\boldsymbol{\Sigma}'}^{-1} & =\frac{1}{\sigma_{w}^{2}}\left(\mathbf{B}_{ij}^{\mathrm{T}}\mathbf{B}_{ij}+\sigma_{w}^{2}\boldsymbol{\Sigma}_{\text{ext},ji}^{-1}\right)\\
\boldsymbol{\mu}' & ={\boldsymbol{\Sigma}'}\left(\boldsymbol{\Sigma}_{\text{ext},ji}^{-1}\boldsymbol{\mu}_{\text{ext},ji}-\frac{1}{\sigma_{w}^{2}}\mathbf{B}_{ij}^{\mathrm{T}}\mathbf{A}_{ij}{\bm{\vartheta}_{i}}\right).
\end{align*}
 If the neighbor $j$ is MN, i.e., $\boldsymbol{\Sigma}_{\text{ext},ji}^{-1}=\text{diag}(\infty,\infty)$,
the parameters reduce to 
\begin{align*}
\boldsymbol{\Sigma}_{\text{in},ij}^{-1} & =\mathbf{A}_{ij}^{\mathrm{T}}\mathbf{A}_{ij}\\
\boldsymbol{\Sigma}_{\text{in},ij}^{-1}\boldsymbol{\mu}_{\text{ext},ji} & =\mathbf{A}_{ij}^{\mathrm{T}}\mathbf{B}_{ij}\boldsymbol{\mu}_{\text{ext},ji}.
\end{align*}
 % xi' (Ai'S^-1 Ai + Ai'S^-1 Aj (Aj'S^-1 Aj + Sj^-1)^-1 Aj S^-1 Ai) xi - 2* xi' Ai' S^-1 Aj (Aj'S^-1 Aj + Sj^-1)^-1 Sj^-1 mu_j + C
% (Sj^-1 mu_j - Aj' S^-1 Ai xi)

For the parameters in \eqref{eq:BPmsg_pij_in} we utilized the generic
formula for multiplying Gaussian normal distributions: 
\begin{align}
\prod_{i}\mathcal{N}_{\mathbf{x}}(\bm{\mu}_{i},\bm{\Sigma}_{i}) & \propto\mathcal{N}_{\mathbf{x}}(\bm{\mu},\bm{\Sigma})\label{eq:ap_GaussMult}
\end{align}
 with $\bm{\Sigma}^{-1}=\sum_{i}\bm{\Sigma}_{i}^{-1}$ and $\bm{\Sigma}^{-1}\bm{\mu}=\sum_{i}\bm{\Sigma}_{i}^{-1}\bm{\mu}_{i}$.

\subsection{Mean Field Update Rules}

\label{ap:MF_update} Derivation of the message parameters for $m_{{p}_{ij}\rightarrow\bm{\vartheta}_{i}}(\bm{\theta}_{i})$
in \eqref{eq:MFmsg_thetai_in}: 
\begin{align*}
m_{{p}_{ij}\rightarrow\bm{\vartheta}_{i}} & (\bm{\vartheta}_{i})=\exp\left(\int\log\left({p}_{ij}(\bm{\vartheta}_{i}, \bm{\vartheta}_{j})\right)\ b_{j}(\bm{\vartheta}_{j})  \, \mathrm{d}\bm{\theta}_{j}' \right)\\
\propto & \exp\Bigg(-\frac{1}{\sigma_{w}^{2}}\int\|\mathbf{A}_{ij}\bm{\vartheta}_{i}+\mathbf{B}_{ij}\bm{\vartheta}_{j}\|^{2} \times \\
& \quad \mathcal{N}_{\bm{\vartheta}_{j}}\left(\boldsymbol{\mu}_{j},\boldsymbol{\Sigma}_{j}\right) \, \mathrm{d}\bm{\theta}_{j}' \Bigg)\\
\propto & \exp\left(-\frac{1}{\sigma_{w}^{2}}\left({\bm{\vartheta}_{i}}^{\mathrm{T}}\mathbf{A}_{ij}^{\mathrm{T}}\mathbf{A}_{ij}\bm{\vartheta}_{i}+{\bm{\vartheta}_{i}}^{\mathrm{T}}\mathbf{A}_{ij}^{\mathrm{T}}\mathbf{B}_{ij}\boldsymbol{\mu}_{j}\right)\right)\\
\propto & \mathcal{N}_{\bm{\vartheta}_{i}}\left(\boldsymbol{\mu}_{\text{in},ij},\boldsymbol{\Sigma}_{\text{in},ij}\right)
\end{align*}
 with the parameters \eqref{eq:MF_Var_thetai_in} and \eqref{eq:MF_Mean_thetai_in}.
The message parameters of \eqref{eq:MF_marg} are derived equivalently
to \eqref{eq:ap_GaussMult}.

\subsection{Convergence proofs}
\label{ap:Conv_proofs}

\begin{figure}
  % priors
%   \psfrag{pt01}[t][t][0.65]{\color[rgb]{0,0,0}\setlength{\tabcolsep}{0pt}\begin{tabular}{c}\raisebox{2.2mm}{$\, p(\bm{\vartheta}_1)$}\end{tabular}}
%   \psfrag{pt02}[t][t][0.65]{\color[rgb]{0,0,0}\setlength{\tabcolsep}{0pt}\begin{tabular}{c}\raisebox{2.2mm}{$\, p(\bm{\vartheta}_2)$}\end{tabular}}
%   \psfrag{pt03}[t][t][0.65]{\color[rgb]{0,0,0}\setlength{\tabcolsep}{0pt}\begin{tabular}{c}\raisebox{2.2mm}{$\, p(\bm{\vartheta}_3)$}\end{tabular}}
%   \psfrag{pt4}[t][t][0.65]{\color[rgb]{0,0,0}\setlength{\tabcolsep}{0pt}\begin{tabular}{c}\raisebox{2.2mm}{$\, p(\bm{\vartheta}_4)$}\end{tabular}}
%   \psfrag{pt5}[t][t][0.65]{\color[rgb]{0,0,0}\setlength{\tabcolsep}{0pt}\begin{tabular}{c}\raisebox{2.2mm}{$\, p(\bm{\vartheta}_5)$}\end{tabular}}
  % likelihoods
%   \psfrag{pt12}[b][b][0.65]{\color[rgb]{0,0,0}\setlength{\tabcolsep}{0pt}\begin{tabular}{c}\vspace{-0.9mm}{$\, \, \tilde{p}(\mathbf{c}_{12} | \bm{\vartheta}_1, \bm{\vartheta}_2)$}\end{tabular}}
  \psfrag{p23c}[l][l][\slab]{\color[rgb]{0,0,0}\setlength{\tabcolsep}{0pt}\begin{tabular}{c}\hspace{-1.1mm}\vspace{0.5mm}{$\phi_{23}$}\end{tabular}}
  \psfrag{p24c}[l][l][\slab]{\color[rgb]{0,0,0}\setlength{\tabcolsep}{0pt}\begin{tabular}{c}\hspace{-1.1mm}\vspace{0.5mm}{$\phi_{24}$}\end{tabular}}
  \psfrag{p25c}[l][l][\slab]{\color[rgb]{0,0,0}\setlength{\tabcolsep}{0pt}\begin{tabular}{c}\hspace{-1.1mm}\vspace{0.5mm}{$\phi_{25}$}\end{tabular}}
  \psfrag{p34c}[l][l][\slab]{\color[rgb]{0,0,0}\setlength{\tabcolsep}{0pt}\begin{tabular}{c}\hspace{-1.1mm}\vspace{0.5mm}{$\phi_{34}$}\end{tabular}}
  \psfrag{p45c}[l][l][\slab]{\color[rgb]{0,0,0}\setlength{\tabcolsep}{0pt}\begin{tabular}{c}\hspace{-1.1mm}\vspace{0.5mm}{$\phi_{45}$}\end{tabular}}
  % variables
  \psfrag{n1}[b][b][\mlab]{\color[rgb]{0,0,0}\setlength{\tabcolsep}{0pt}\begin{tabular}{c}\vspace{-1.2mm}\hspace{-0.5mm}{$\bm{\vartheta}_1 \ist$}\end{tabular}}
  \psfrag{n2}[b][b][\mlab]{\color[rgb]{0,0,0}\setlength{\tabcolsep}{0pt}\begin{tabular}{c}\vspace{-1.2mm}\hspace{-0.5mm}{$\bm{\vartheta}_2$}\end{tabular}}
  \psfrag{n3}[b][b][\mlab]{\color[rgb]{0,0,0}\setlength{\tabcolsep}{0pt}\begin{tabular}{c}\vspace{-1.2mm}\hspace{-0.5mm}{$\bm{\vartheta}_3 \rmv$}\end{tabular}}
  \psfrag{n4}[b][b][\mlab]{\color[rgb]{0,0,0}\setlength{\tabcolsep}{0pt}\begin{tabular}{c}\vspace{-1.2mm}\hspace{-0.5mm}{$\bm{\vartheta}_4 \rmv$}\end{tabular}}
  \psfrag{n5}[b][b][\mlab]{\color[rgb]{0,0,0}\setlength{\tabcolsep}{0pt}\begin{tabular}{c}\vspace{-1.2mm}\hspace{-0.5mm}{$\bm{\vartheta}_5 \rmv$}\end{tabular}}
  
  \psfrag{la}[b][b][\mlab]{\color[rgb]{0,0,0}\setlength{\tabcolsep}{0pt}\begin{tabular}{c}{(a)}\end{tabular}}
  \psfrag{lb}[b][b][\mlab]{\color[rgb]{0,0,0}\setlength{\tabcolsep}{0pt}\begin{tabular}{c}{(b)}\end{tabular}}
  %figure
  \centering{}\includegraphics[width=1\columnwidth]{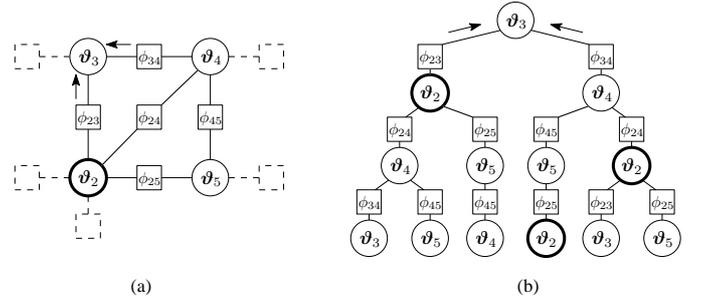}
  \caption{(a) FG and (b) its computation tree to vertex 3 after 3 iterations.}
  \label{fig:Comp_Tree}
\end{figure}

\highlight{
For the convergence proofs, we first introduce the concept of computation trees on the example of the FG from Fig.~\ref{fig:FG_network}. As the MN represented by vertex $\bm{\vartheta}_1$ is a fixed parameter, it can be considered as prior information to vertex $\bm{\vartheta}_2$. As function vertices that are connected only to a single variable vertex can be merged\footnote{While in general this merging is not unique, for our purpose it is sufficient to
divide the prior into equal parts for each connected likelihood function.} into pairwise function vertices, Fig.~\ref{fig:Comp_Tree} (a) is an equivalent representation where $\phi_{ij}$ comprises the prior and the local likelihood. The message passing in the cyclic graph is equivalent to the message passing in a computation tree, where the considered variable vertex represents the root of the tree. In Fig.~\ref{fig:Comp_Tree} (b), the computation tree to vertex $\bm{\vartheta}_3$ for three message passing iterations is depicted, i.e., the computation tree of depth 3. In the tree, only messages towards the root vertex are considered. For later use, also the variable vertex $\bm{\vartheta}_2$ is highlighted, as it is connected to the master node.
}

\highlight{
 \subsubsection{Proof of  Theorem \ref{pr:Conv_prior}} \label{sec:proof1}
%\begin{proof}
The precision matrix of any pairwise factor between vertex $i$ and $j$ can be written as $\mathbf{J}_{ij} \triangleq \bm{\Sigma}_{ij}^{-1} = [\mathbf{A}_{ij}, \mathbf{B}_{ij}]^\mathrm{T} [\mathbf{A}_{ij}, \mathbf{B}_{ij}]$, which is positive semidefinite (p.s.d.) because of the full correlation of the second column in $\mathbf{A}_{ij}$ and $\mathbf{B}_{ij}$.
Without loss of generality, we can distribute the prior information to the pairwise components as
      %\begin{align*}
	    $\mathbf{J}_{ij}' = \mathbf{J}_{ij}
	    +
	    \bm{\Sigma}_{\text{p},i}^{-1} / |\mathcal{T}_i|$, 
      %\end{align*}
      where $\bm{\Sigma}_{\text{p},i}^{-1} \succ 0$ denotes the precision of the prior information. As all $\mathbf{J}_{ij}' \succ 0$, the system is factor graph normalizable \cite[Prop. 4.3.3]{malioutov}, and thus the variances of BP are guaranteed to converge.
%\end{proof}
}

\highlight{
\subsubsection{Proof of Theorem \ref{pr:Conv_master}} \label{sec:proof2}
%\begin{proof}[Proof of theorem \ref{pr:Conv_master}]
In this case we cannot guarantee positive definite (p.d.) precision matrices in the factor vertices. The following proof is similar to the proof of \cite[Prop. 4.3.3]{malioutov}.
In the computation tree, paths with a connection to a MN appear. We will show that these paths have a p.d. contribution to the marginal of the root vertex. 
%In the computation tree, we will show that from paths where ANs are connected to a MN originate p.d. shares to the precision of the marginal in the root vertex. 
Based on this result, we can show that the variances are bounded from below and monotonically decreasing. In the following, derivations the scaling with ${1}/{\sigma_w^2}$ is dropped.\\
  \emph{Step 1} (Path marginals)
    Consider a vertex $j$ connected to a MN $k$. The MN adds the p.d. main diagonal entry $\mathbf{A}_{jk}^\mathrm{T} \mathbf{A}_{jk}$ to the covariance of $j$ (see \eqref{eq:MAP_covar}), which can be equally distributed to the adjacent factors. Now consider a path $h \leftarrow i \leftarrow j$, where the leaf vertex $j$ has a connection to a MN $k$. The precision matrix
    %writes as
    %If a vertex $i$ is connected to a neighboring vertex $j \in \mathcal{T}_i$, which communicates to a MN $k$, the precision of pairwise factor writes as
    is given by
	   \begin{align*}
	      \bm{\Sigma}_{hij}^{-1} &= \begin{bmatrix} 
		  \mathbf{A}_{hi}^\mathrm{T} \mathbf{A}_{hi} & \mathbf{A}_{hi}^\mathrm{T} \mathbf{B}_{hi} & \mathbf{0}\\
		  \mathbf{B}_{hi}^\mathrm{T} \mathbf{A}_{hi} & \mathbf{B}_{hi}^\mathrm{T} \mathbf{B}_{hi} + \mathbf{A}_{ij}^\mathrm{T} \mathbf{A}_{ij} & \mathbf{A}_{ij}^\mathrm{T} \mathbf{B}_{ij} \\
		  \mathbf{0} & \mathbf{B}_{ij}^\mathrm{T} \mathbf{A}_{ij} & \mathbf{B}_{ij}^\mathrm{T} \mathbf{B}_{ij}  + \mathbf{Q}_{jk}
		\end{bmatrix},
	  \end{align*}
	  with $\mathbf{Q}_{jk} = 1/|\mathcal{T}_i| \, \mathbf{A}_{jk}^\mathrm{T} \mathbf{A}_{jk} \succ 0$.
	  We marginalize out the leaf vertex $j$ and obtain the precision matrix of $h$ and $i$
      	  \begin{align*}
	      {\bm{\Sigma}}_{hi,j}^{-1} &= \begin{bmatrix} 
		  \mathbf{A}_{hi}^\mathrm{T} \mathbf{A}_{hi} & \mathbf{A}_{hi}^\mathrm{T} \mathbf{B}_{hi} \\
		  \mathbf{B}_{hi}^\mathrm{T} \mathbf{A}_{hi} & \mathbf{B}_{hi}^\mathrm{T} \mathbf{B}_{hi} + \mathbf{S}_{ij}
		\end{bmatrix},
	  \end{align*}      
      where $\mathbf{S}_{ij}$ is the Schur complement
      \begin{align*}
	    {\mathbf{S}}_{ij} = & \mathbf{A}_{ij}^\mathrm{T} \mathbf{A}_{ij} - \mathbf{A}_{ij}^\mathrm{T} \mathbf{B}_{ij}  
	    (\mathbf{B}_{ij}^\mathrm{T} \mathbf{B}_{ij} + \mathbf{Q}_{jk})^{-1}
	    \mathbf{B}_{ij}^\mathrm{T} \mathbf{A}_{ij} \nonumber \\
	    = &
	    \underbrace{\mathbf{A}_{ij}^\mathrm{T} \mathbf{A}_{ij} - \mathbf{A}_{ij}^\mathrm{T} \mathbf{B}_{ij} \mathbf{B}_{ij}^+ \mathbf{A}_{ij}}_
	    {\succeq 0}\\
	    &
	    + 
	    \underbrace{\mathbf{A}_{ij}^\mathrm{T} (\mathbf{B}_{ij}^+)^\mathrm{T} (\mathbf{Q}_{jk}^{-1} + (\mathbf{B}_{ij}^\mathrm{T} \mathbf{B}_{ij})^{-1})^{-1} \mathbf{B}_{ij}^+ \mathbf{A}_{ij}}
	    _{\succ 0}.
       \end{align*}
      The expansion from the first to the second line uses the Woodbury identity. In the same way the positive definiteness propagates until the root. Thus, paths including MNs always have a p.d. marginal precision, whereas paths not including any MN connection have a p.s.d. marginal precision.  \\
  \emph{Step 2} (Bounded from below) As we require at least one MN in the tree, at least one path will have a p.d. share to the marginal precision matrix of the root vertex. Thus, the marginal of the root vertex has upper bounded precision, and lower bounded covariance.
    \\
  \emph{Step 3} (Monotonic decreasing) Every time when the depth of the computation tree increases from $n$ to $n+1$, and another vertex connected to a MN is added as leaf, an additional p.d.~share is added to the corresponding path. Thus, the precision of the root marginal increases in the p.d.~sense, and its covariance (the inverse) in the negative definite sense. In summary, variances decrease monotonically as a the depth of the computation tree increases, and as they are bounded from below, hence they converge.
%\end{proof}
}
%--------------------------------------------------------------------
% SECTION: References
%--------------------------------------------------------------------
%\bibliographystyle{IEEEtran}
\bibliographystyle{ieeetr_noParentheses}
\bibliography{references}

\begin{thebibliography}{10}

\bibitem{jagannathan}
S.~Jagannathan, H.~Aghajan, and A.~Goldsmith, ``{The effect of time
  synchronization errors on the performance of cooperative MISO systems},'' in
  {\em IEEE Global Commun.\ Conf. Workshops 2004}, pp.~102 -- 107, Nov. 2004.

\bibitem{demirkol}
I.~Demirkol, C.~Ersoy, and F.~Alagoz, ``{MAC protocols for wireless sensor
  networks: a survey},'' {\em IEEE Commun.\ Mag.}, vol.~44, pp.~115 -- 121,
  Apr. 2006.

\bibitem{ganeriwal05}
S.~Ganeriwal, D.~Ganesan, H.~Shim, V.~Tsiatsis, and M.~B. Srivastava,
  ``Estimating clock uncertainty for efficient duty-cycling in sensor
  networks,'' in {\em Proc. 3rd int. Conf. on Emb. networked sensor sys.},
  SenSys '05, New York, NY, USA, pp.~130--141, ACM, 2005.

\bibitem{hlinka}
O.~Hlinka, F.~Hlawatsch, and P.~M. Djuric, ``Distributed particle filtering in
  agent networks: A survey, classification, and comparison.,'' {\em IEEE Signal
  Process.\ Mag.}, vol.~30, pp.~61--81, Jan. 2013.

\bibitem{elson2}
J.~Elson and K.~R\"{o}mer, ``Wireless sensor networks: a new regime for time
  synchronization,'' {\em SIGCOMM Comput. Commun. Rev.}, vol.~33, pp.~149--154,
  Jan. 2003.

\bibitem{antonelli13}
G.~Antonelli, ``Interconnected dynamic systems: An overview on distributed
  control,'' {\em Control Systems, IEEE}, vol.~33, no.~1, pp.~76--88, 2013.

\bibitem{simeone}
O.~Simeone, U.~Spagnolini, Y.~Bar-Ness, and S.~H. Strogatz, ``Distributed
  synchonization in wireless networks,'' {\em IEEE Signal Process.\ Mag.},
  vol.~25, pp.~81--97, Sept. 2008.

\bibitem{wu}
Y.-C. Wu, Q.~M. Chaudhari, and E.~Serpedin, ``Clock synchronization of wireless
  sensor networks,'' {\em IEEE Signal Process.\ Mag.}, vol.~28, pp.~124--138,
  Jan. 2011.

\bibitem{elson}
J.~Elson, L.~Girod, and D.~Estrin, ``Fine-grained network time synchronization
  using reference broadcasts,'' in {\em Proc. 5th Symp. Operat. Syst. Design
  Implement.}, pp.~147--163, Dec. 2002.

\bibitem{maroti}
M.~Mar\'{o}ti, B.~Kusy, G.~Simon, and A.~L{\'e}deczi, ``The flooding time
  synchronization protocol,'' in {\em Proc. 2nd Int. Conf. on Embedded
  networked sensor systems}, New York, NY, USA, pp.~39--49, ACM, Nov. 2004.

\bibitem{zennaro}
D.~Zennaro, E.~Dall'Anese, T.~Erseghe, and L.~Vangelista, ``{Fast clock
  synchronization in wireless sensor networks via ADMM-based consensus},'' in
  {\em Proc. 9th Int. Symp. Model. Optim. Mobile, Ad Hoc, Wireless Netw.},
  pp.~148--153, May 2011.

\bibitem{schenato}
L.~Schenato and F.~Fiorentin, ``Average timesynch: A consensus-based protocol
  for clock synchronization in wireless sensor networks,'' {\em Automatica},
  vol.~47, pp.~1878 -- 1886, Sept. 2011.

\bibitem{maggs}
M.~Maggs, S.~O'Keefe, and D.~Thiel, ``Consensus clock synchronization for
  wireless sensor networks,'' {\em IEEE Sensors J.}, vol.~12, no.~6,
  pp.~2269--2277, 2012.

\bibitem{zhao}
D.~Zhao, Z.~An, and Y.~Xu, ``Time synchronization in wireless sensor networks
  using max and average consensus protocol,'' {\em Int. J. Dist. Sens. Netw.},
  vol.~2013, 2013.

\bibitem{solis}
R.~Solis, V.~Borkar, and P.~Kumar, ``A new distributed time synchronization
  protocol for multihop wireless networks,'' in {\em Proc. 45th IEEE Conf.
  Decis. Control}, pp.~2734 --2739, Dec. 2006.

\bibitem{leng}
M.~Leng and Y.-C. Wu, ``Distributed clock synchronization for wireless sensor
  networks using belief propagation,'' {\em IEEE Trans.\ Signal Process.},
  vol.~59, pp.~5404--5414, Nov. 2011.

\bibitem{zennaro13}
D.~Zennaro, A.~Ahmad, L.~Vangelista, E.~Serpedin, H.~N. Nounou, and M.~N.
  Nounou, ``Network-wide clock synchronization via message passing with
  exponentially distributed link delays,'' {\em IEEE Trans.\ Commun.}, vol.~61,
  no.~5, pp.~2012--2024, 2013.

\bibitem{du13}
J.~Du and Y.-C. Wu, ``Fully distributed clock skew and offset estimation in
  wireless sensor networks,'' {\em Proc. IEEE Int.\ Conf.\ Acoust.,\ Speech,\
  Sig.\ Process.}, pp.~1--5, June 2013.

\bibitem{chepuri13}
S.~P. Chepuri, R.~T. Rajan, G.~Leus, and van der Alle-Jan van~der Veen, ``Joint
  clock synchronization and ranging: Asymmetrical time-stamping and passive
  listening,'' {\em IEEE Signal Process.\ Lett.}, vol.~20, no.~1, pp.~51--54,
  2013.

\bibitem{wu2}
J.~Zheng and Y.-C. Wu, ``Joint time synchronization and localization of an
  unknown node in wireless sensor networks,'' {\em IEEE Trans.\ Signal
  Process.}, vol.~58, pp.~1309--1320, Mar. 2010.

\bibitem{noh}
K.-L. Noh, Q.~M. Chaudhari, E.~Serpedin, and B.~W. Suter, ``Novel clock phase
  offset and skew estimation using two-way timing message exchanges for
  wireless sensor networks,'' {\em IEEE Trans.\ Commun.}, vol.~55,
  pp.~766--777, Apr. 2007.

\bibitem{loschmidt}
P.~Loschmidt, R.~Exel, A.~Nagy, and G.~Gaderer, ``Limits of synchronization
  accuracy using hardware support in ieee 1588,'' in {\em IEEE Int. Symp. Prec.
  Clock Synch.}, pp.~12 --16, Sept. 2008.

\bibitem{cristian89}
F.~Cristian, ``Probabilistic clock synchronization,'' {\em Distributed
  Computing}, vol.~3, pp.~146 -- 158, 1989.

\bibitem{trees}
H.~L. van Trees, {\em Detection, Estimation, and Modulation Theory: Radar-Sonar
  Signal Processing and Gaussian Signals in Noise}.
\newblock Melbourne, FL, USA: Krieger Publishing Co., Inc., 1992.

\bibitem{wymeersch}
H.~Wymeersch, J.~Lien, and M.~Z. Win, ``Cooperative localization in wireless
  networks,'' {\em Proc.\ IEEE}, vol.~97, pp.~427--450, Feb. 2009.

\bibitem{yedidia}
J.~S. Yedidia, W.~T. Freeman, and Y.~Weiss, ``Constructing free energy
  approximations and generalized belief propagation algorithms,'' {\em IEEE
  Trans.\ Inf.\ Theory}, vol.~51, pp.~2282--2312, July 2005.

\bibitem{kschischang}
F.~Kschischang, B.~Frey, and H.-A. Loeliger, ``Factor graphs and the
  sum-product algorithm,'' {\em IEEE Trans.\ Inf.\ Theory}, vol.~47, pp.~498
  --519, Feb. 2001.

\bibitem{dauwels}
J.~Dauwels, ``On variational message passing on factor graphs,'' {\em Proc.
  IEEE Int. Symp. Inf. Theory}, pp.~2546 --2550, June 2007.

\bibitem{koller09}
D.~Koller and N.~Friedman, {\em Probabilistic Graphical Models: Principles and
  Techniques}.
\newblock MIT Press, 2009.

\bibitem{wainwright}
M.~J. Wainwright and M.~I. Jordan, ``Graphical models, exponential families,
  and variational inference,'' {\em Found. Trends Mach. Learn.}, vol.~1,
  pp.~1--305, Jan. 2008.

\bibitem{weiss}
Y.~Weiss and W.~T. Freeman, ``Correctness of belief propagation in gaussian
  graphical models of arbitrary topology.,'' {\em Neural Computation}, vol.~13,
  pp.~2173--2200, Oct. 2001.

\bibitem{malioutov06}
D.~M. Malioutov, J.~K. Johnson, and A.~S. Willsky, ``Walk-sums and belief
  propagation in gaussian graphical models,'' {\em J. Mach. Learn. Res.},
  vol.~7, pp.~2031--2064, 2006.

\bibitem{malioutov}
D.~M. Malioutov, ``Approximate inference in gaussian graphical models.''
  \textit{Ph.D. Thesis, Dept. Elect. Eng. Comp. Sc., MIT}, 2008.

\end{thebibliography}

\end{document}